\documentclass[11pt]{article}

\usepackage{amssymb} 
\usepackage{amsmath}     
\usepackage{amsthm}
\usepackage{todonotes}
\usepackage{mathtools}
\usepackage{a4wide}
\usepackage{graphicx}
\usepackage{hyperref}

\usepackage[ruled,vlined,linesnumbered]{algorithm2e}
\SetKw{KwNot}{not}
\SetKw{KwOR}{or}
\SetKw{KwAND}{and}
\SetKw{KwExitFor}{exit for-loop}
\SetKw{KwStep}{Step}

\usepackage{authblk}

\newcommand{\Rset}{\mathbb{R}}
\newcommand{\Xset}{\mathbb{X}}

\newcommand{\Uset}{\mathbb{U}}

\newcommand{\SetOfClauses}{\mathcal{C}}

\newtheorem{thm}{Theorem}
\newtheorem{lem}{Lemma}

\newtheorem{cor}{Corollary}

\begin{document}

\title{Recoverable  robust shortest path problem under interval budgeted uncertainty representations} 

\author[1]{Marcel Jackiewicz}
\author[1]{Adam Kasperski\footnote{Corresponding author}}
\author[1]{Pawe{\l} Zieli\'nski}

\affil[1]{
Wroc{\l}aw  University of Science and Technology, Wroc{\l}aw, Poland\\
            \texttt{\{marcel.jackiewicz,adam.kasperski,pawel.zielinski\}@pwr.edu.pl}}


\maketitle

 \begin{abstract}
This paper deals with the recoverable robust shortest path problem under interval uncertainty representations. In this problem, a first-stage path is computed, which can be modified to some extent after observing changes in the cost structure. The uncertain second-stage arc costs are modeled by intervals, and the robust min-max criterion is used to compute an optimal solution.
The problem is known to be strongly NP-hard and also hard to approximate in general digraphs. However, until now its complexity for acyclic digraphs was unknown. In this paper, it is shown that the problem in acyclic digraphs can be solved in polynomial time for the traditional interval uncertainty and all natural neighborhoods known from the literature. More efficient algorithms for layered and arc series-parallel digraphs are constructed. Hardness results for general digraphs are also strengthened. Finally, some exact and approximate methods of solving the problem under interval budgeted uncertainty are proposed.
  \end{abstract}

\section{Introduction}

The concept of recoverable robustness was introduced by Liebchen et al.~\cite{LLMS09}. This two-stage approach consists in computing a first-stage solution whose cost is known. In the second stage, after the uncertain costs are revealed, a limited recovery action is allowed to modify the first-stage solution. We seek a solution whose total first and second-stage cost is minimal. The recoverable robust approach can be naturally applied to the class of combinatorial optimization problems, where each solution can be represented as a subset of a finite element set. In this case, the recovery action consists in adding some elements to the first-stage solution or excluding some elements from it (see, e.g.,~\cite{SAO09, NO13,IKKKO22}). The second-stage uncertainty can be modeled in various ways. We can use the \emph{discrete uncertainty representation}
(see, e.g.,~\cite{KY97}) by simply listing all possible realizations of the second-stage costs. 
Alternatively, we can use the \emph{interval uncertainty representation} by providing an interval of possible values for each second-stage cost. This representation is convenient in applications because it requires only a nominal value and a maximum deviation from this nominal value for each uncertain parameter.
To control the amount of uncertainty and the price of robustness, budgeted versions of the interval uncertainty can be used~\cite{BS03, BS04, NO13}. A budget allows us to avoid over-conservatism of the solutions computed. Indeed, it can be unlikely that all the uncertain parameters take their worst values simultaneously. A deeper motivation for using budgets in the interval uncertainty can be found, for example, in~\cite{BS03}. 

The recoverable robust approach has recently been applied to various combinatorial optimization problems. Under the interval uncertainty representation, the problem can be solved in polynomial time for the class of matroidal problems~\cite{LPT19}, for example, the selection~\cite{KZ15b, LLW21} and minimum spanning tree~\cite{HKZ16, HKZ16a} problems. On the other hand, the problem is strongly NP-hard for the shortest path in general digraphs~\cite{B12} and assignment~\cite{FHLW21} problems. Furthermore, the former problem is also hard to approximate even if a very limited recovery action is allowed~\cite{B12}, while the latter one is $W[1]$-hard with respect to the recovery parameter~\cite{FHLW21}. In~\cite{CG15b, GLW21}, the recoverable robust version of the strongly NP-hard traveling salesperson problem under interval uncertainty has been investigated. Several approximation algorithms for some special cases of this problem were proposed. In~\cite{BG22}, the recoverable robust version of a single-machine scheduling problem with interval job processing times was discussed. Some approximation algorithms for this problem were proposed.

Adding budgets to the interval uncertainty makes solving the recoverable robust problem challenging. Polynomial-time algorithms are known only for some very special cases, such as the selection problem under the 
\emph{continuous budgeted uncertainty}~\cite{NO13}, constructed in~\cite{CGKZ18}.  The complexity of the analogous problem for the minimum spanning tree remains open. For the discrete budgeted uncertainty~\cite{BS04}, a compact mixed integer programming model for the recoverable robust selection problem can be constructed~\cite{CGKZ18}. For the analogous recoverable robust version of the minimum spanning tree problem, computing a worst scenario for a given first-stage solution (the adversarial problem) is already strongly NP-hard~\cite{NO13}. In~\cite{BGKK19, BKK11}, the recoverable robust knapsack problem with
 the budgeted uncertainty in the item weights was investigated. A compact mixed integer programming formulation for this problem was presented. Generally, the recoverable robust combinatorial problems
under budgeted interval uncertainty can be solved using a row generation algorithm described in~\cite{HKZ19, ZZ13}. However, this algorithm is efficient only for problems of small size. Designing more efficient methods for particular problems is still very challenging.

In this paper, we focus on the recoverable, robust version of the shortest path problem. This problem was first investigated in~\cite{B11, B12}, and a closely related problem, called the incremental one, was also previously discussed in~\cite{SAO09}. In the traditional shortest path problem, we are given a directed graph with arc costs, and we seek a simple shortest path between two specified nodes. This well-known network problem can be solved in polynomial time in general digraphs, assuming that the arc costs are nonnegative, and in acyclic digraphs for arbitrary costs (see, e.g.,~\cite{AMO93}). The complexity of its recoverable robust version can be much worse. The paths resulting from a recovery action form the so-called neighborhood of a given first-stage path. It turns out that even computing an optimal recovery action for a given second-stage scenario is strongly NP-hard for some natural neighborhoods~\cite{SAO09, NO13}.
Furthermore, computing the optimal first and second-stage paths is strongly NP-hard for all natural neighborhoods~\cite{B12,NO13}.  This is, however, the case in general digraphs. Until now, the complexity of the problem for acyclic digraphs has been open. In this paper, we close this gap and show that the recoverable robust shortest path problem in acyclic multidigraphs,  under interval uncertainty, can be solved in polynomial time. This fact remains valid for all neighborhoods discussed in the literature. We construct polynomial time algorithms for general acyclic digraphs and show that the running time can be improved for layered and arc series-parallel digraphs. We also strengthen the hardness results for general digraphs. Namely, we show that the recoverable robust shortest path problem is strongly NP-hard and hard to approximate for digraphs which are near acyclic planar ones. In the second part of the paper, we discuss the robust version of the problem with budgeted interval uncertainty. We show that the problem with the continuous budgeted uncertainty can be solved by using a compact mixed integer programming formulation. Furthermore, the polynomial time algorithms for acyclic digraphs and the interval uncertainty allow us to construct efficient approximation algorithms for both continuous and discrete budgeted uncertainty.

This paper is organized as follows. In Section~\ref{secdef}, we recall the formulation of the recoverable robust shortest path problem with various neighborhoods and several interval uncertainty representations. We also recall the definitions of the inner adversarial and incremental problems.  In Section~\ref{secinterv}, we discuss the case of the traditional interval uncertainty. We strengthen the known complexity results for general digraphs and show that the problem can be solved in polynomial time for acyclic multidigraphs. We present several polynomial time algorithms for this class of graphs. In Section~\ref{secrobust}, we discuss the recoverable robust shortest path problem with the budgeted interval uncertainty. We construct a compact mixed integer programming formulation for the continuous case while strengthening some known complexity results for the discrete case. Finally, we propose several approximation algorithms for the class of acyclic multidigraphs.

\section{Recoverable  robust shortest path problem }
\label{secdef}

In the \emph{shortest path problem}  (\textsc{SP})
we are given a \emph{multidigraph} $G=(V,A)$ that
 consists
of a finite set of nodes~$V$, $|V|=n$,  and a finite multiset of
arcs~$A$, $|A|=m$. Two nodes  $s\in V $ and $t\in V$ are
distinguished as the \emph{starting node} and the
\emph{destination} node, respectively. In particular, node~$s$ is
called \textit{source} if no arc enters $s$, and node~$t$ is called
\textit{sink} if no arc leaves~$t$. Let $\Phi$ be the set of all simple $s$-$t$ paths in~$G$. We will identify each path $X\in \Phi$ with the corresponding set of arcs that form $X$.  
  A deterministic cost is associated  with each arc of $G$, and  
we seek a simple $s$-$t$ path in $G$        
with the minimum total cost. The shortest path problem can be solved efficiently using several polynomial-time algorithms (see, e.g.,~\cite{AMO93}). 

In the \emph{recoverable robust} version of~\textsc{SP} (\textsc{Rec Rob SP} for short), 
 we are given first-stage arc costs $C_e\geq 0$, $e\in A$, which are known in advance. On the other hand, the second-stage arc costs are uncertain, and 
their particular realization $S=(c_e^S)_{e\in A}$ is called a \emph{scenario}. The set of all possible scenarios is specified as a \emph{scenario (uncertainty) set}  $\Uset\subset \Rset_{+}^{|A|}$. 
 The decision process in   \textsc{Rec Rob SP} is two-stage and consists in choosing a path~$X\in \Phi$
 in the first stage. Then, in the second stage (the recovery stage), after a scenario $S\in \Uset$ is revealed, 
 the path~$X$ can be modified to some extent. This modification consists in finding a shortest path~$Y$ under the cost scenario~$S$
 in a \emph{neighborhood} of~$X$, denoted by $\Phi(X,k)\subseteq \Phi$. This neighborhood depends on a  given
 \emph{recovery parameter}~$k$, being an integer such that $0\leq k\leq |A|$. In this paper, we examine the following variants of the neighborhood $\Phi(X,k)$ (see, e.g.,~\cite{SAO09,NO13}):
\begin{align}
\Phi^{\mathrm{incl}}(X,k)&=\{Y\in\Phi \,:\,  |Y\setminus X|\leq k\}, \label{incl}\\
\Phi^{\mathrm{excl}}(X,k)&=\{Y\in\Phi \,:\,  |X\setminus Y|\leq k\}, \label{excl}\\
\Phi^{\mathrm{sym}}(X,k)&=\{Y\in\Phi \,:\,  | (Y\setminus X) \cup (X\setminus Y) |\leq k\} \label{sym}
\end{align}
called  the \emph{arc inclusion  neighborhood} (at most $k$ new arcs can be added to $X$),  the \emph{arc exclusion  neighborhood} (at most $k$ arcs can be removed from $X$), and   the \emph{arc symmetric difference  neighborhood} (at most $k$ arcs can be different in $X$ and $Y$), respectively. 

The goal in \textsc{Rec Rob SP}  is to find a first-stage path $X\in \Phi$ and a second-stage path $Y\in \Phi(X,k)$ which minimize the total cost in a worst second-stage cost scenario $S\in \mathbb{U}$.
 Therefore, the problem can be stated as follows:
 \begin{equation}
\textsc{Rec Rob SP} : \; \min_{X\in \Phi} \left(\sum_{e\in X} C_e + \max_{S\in \Uset}\min_{Y\in \Phi(X,k)} \sum_{e\in Y} c_e^S\right).
\label{rrsp}
\end{equation}

The \textsc{Rec Rob SP}  problem contains two inner problems.
The first one is  the \emph{adversarial problem}, in which an  \emph{adversary} wants to find, for  a given first-stage path~$X\in \Phi$, a scenario
that  leads to the greatest increase in  the cost of a shortest path from~$\Phi(X,k)$:
\begin{equation}
\textsc{Adv SP}: \; \max_{S\in \Uset}\min_{Y\in \Phi(X,k)} \sum_{e\in Y} c_e^S.
\label{adv}
\end{equation}
The \textsc{Rec Rob SP}  problem reduces to \textsc{Adv SP} when we set $C_e=0$ if $e\in X$ and $C_e=M$ otherwise,
where $M$ is a sufficiently large number.
Indeed, the fixed  path $X$ is then the only reasonable choice in the first stage.

In the second inner problem, called the \emph{incremental problem}, the goal is to make some modifications of~$X$ that consist in finding a cheapest path~$Y\in \Phi(X,k)$ under the cost scenario~$S$
in order to adjust~$X$ to the cost realization:
\begin{equation}
\textsc{Inc SP}: \; \min_{Y\in \Phi(X,k)} \sum_{e\in Y} c_e^{S},
\label{inc}
\end{equation}
where  path~$X\in \Phi$ is given, and the uncertain second-stage arc costs are realized in the form of the scenario~$S$.
This problem models the decision making in the second stage.
The \textsc{Rec Rob SP}  problem reduces to \textsc{Inc SP} if
$\Uset=\{S\}$, $C_e=0$ if $e\in X$ and $C_e=M$, otherwise. Observe that the hardness of the adversarial (or incremental) problem immediately implies the hardness of \textsc{Rec Rob SP}.

Let  $\hat{c}_e$ be a \emph{nominal} second-stage cost of arc $e\in A$ and let $\Delta_e\geq 0$ be the maximum 
deviation of the second-stage cost of~$e$ from its nominal value. In this paper, we use the \emph{interval uncertainty}, so we assume that $c_e^S\in [\hat{c}_e, \hat{c}_e+\Delta_e]$ for each $e\in A$.
We will consider three particular cases of the interval uncertainty,  namely $\Uset\in \{\mathcal{U}, \mathcal{U}(\Gamma^d), \mathcal{U}(\Gamma^c)\}$, where

\begin{align}
& \mathcal{U} =\{S=(c_e^S)_{e\in A}\,:\, c_e^S\in [\hat{c}_e, \hat{c}_e+\Delta_e], e\in A\},
\label{intset}\\
& \mathcal{U}(\Gamma^d)=\{S=(c_e^S)_{e\in A}\in \mathcal{U} \,:\, |\{e\in A\,:\, c_e^S>\hat{c}_e\}|\leq \Gamma^d \},
\label{intsetgd}\\
& \mathcal{U}(\Gamma^c)=\{S=(c_e^S)_{e\in A}\in \mathcal{U} \,:\, \sum_{e\in A} (c_e^S-\hat{c}_e) \leq \Gamma^c\}.
\label{intsetgc}
\end{align}
Notice that $\mathcal{U}$ is the traditional interval uncertainty set, being the Cartesian product of the uncertainty intervals. The sets $\mathcal{U}(\Gamma^d)$ and $\mathcal{U}(\Gamma^c)$ are budgeted versions of $\mathcal{U}$, where  $\mathcal{U}(\Gamma^d)$ is called \emph{discrete budgeted uncertainty}~\cite{BS03,BS04} and 
$\mathcal{U}(\Gamma^c)$ is 
called \emph{continuous budgeted uncertainty}~\cite{NO13}. 
The parameters $\Gamma^d\in  \{0,\ldots,|A|\}$  and $\Gamma^c\in \Rset_+$ are called \emph{budgets} and allow us to control the amount of uncertainty in $\mathcal{U}$. The parameter $\Gamma^d$ limits the number of second-stage costs that can deviate from their nominal values. On the other hand, the parameter $\Gamma^c$ limits the total deviation of the second-stage costs from their nominal values. 
In the following we will use $\Uset\in \{\mathcal{U}, \mathcal{U}(\Gamma^d), \mathcal{U}(\Gamma^c)\}$ in~(\ref{rrsp})--(\ref{inc}). Observe that $\mathcal{U}$ is a special case of $\mathcal{U}(\Gamma^d)$ and $\mathcal{U}(\Gamma^c)$ when the budgets are sufficiently large.

If the recovery parameter $k=0$, then the first and the second-stage paths must be identical.
The \textsc{Rec Rob SP}  problem is then equivalent to the min-max shortest path problem (\textsc{MinMax SP}) in which we seek a path minimizing the largest cost in scenario sets $\mathcal{U}'$, $\mathcal{U}'(\Gamma^c)$, $\mathcal{U}'(\Gamma^d)$, where $\hat{c}_e'=C_e+\hat{c}_e$ and  $\Delta_e'=\Delta_e$ for each $e\in A$.
The \textsc{MinMax SP} problem can be solved in polynomial time by solving a family of deterministic shortest paths problems (see, e.g.,~\cite{BS03,LK14}).
From now on, we will assume that the recovery parameter~$k\geq 1$. 

\section{The \textsc{Rec Rob SP}  problem under the interval uncertainty}
\label{secinterv}

In this section, we consider the  \textsc{Rec Rob SP}   problem under the interval uncertainty $\mathcal{U}$.
We begin by observing that, in this case,
the adversarial problem~(\ref{adv})  can be simplified.
Indeed, it is easily seen that
\[
\max_{S\in \mathcal{U}}\min_{Y\in \Phi(X,k)} \sum_{e\in Y} c_e^S=
\min_{Y\in \Phi(X,k)} \sum_{e\in Y} (\hat{c}_e+\Delta_e).
\]
 Thus,  \textsc{Adv SP} becomes the \textsc{Inc SP} problem under the upper bound scenario $(\hat{c}_e+\Delta_e)_{e\in A}$.
 Hence, the \textsc{Rec Rob SP}  under~$\mathcal{U}$ also simplifies and  is equivalent to
 the following problem:
 \begin{equation}
\textsc{Rec SP}: \; \min_{X\in \Phi}  \left(\sum_{e\in X} C_e+\min_{Y\in \Phi(X,k)} \sum_{e\in Y} \overline{c}_e\right)
=
\min_{X\in \Phi, Y\in \Phi(X,k)} \left(\sum_{e\in X} C_e +  \sum_{e\in Y} \overline{c}_e\right),
\label{rsp}
\end{equation}
 called the
\emph{recoverable   shortest path problem}, where  $\overline{c}_e$ stands for $\hat{c}_e+\Delta_e$ for each $e\in A$.
Thus, throughout this section, we will study the \textsc{Rec SP} problem instead of 
 the \textsc{Rec Rob SP}  under~$\mathcal{U}$ and the \textsc{Inc SP} problem instead of~\textsc{Adv SP}.

 Table~\ref{tab1} summarizes the results obtained later in this section. We will first strengthen the known hardness results for general digraphs. We then construct polynomial time algorithms for various classes of acyclic multidigraphs.
 \begin{table}[h]
 \begin{small}
 \caption{Summary of the results for \textsc{Rec SP} 
 under the interval uncertainty $\mathcal{U}$
 for various classes of multidigraphs and neighborhoods.}
 \label{tab1}
 \begin{center}
 \begin{tabular}{l|l|l|l}
                      \hline
                    &\multicolumn{3}{c}{Neighborhoods}\\
                     \cline{2-4}
Multidigraph& $\Phi^{\mathrm{incl}}(X,k)$ & $\Phi^{\mathrm{excl}}(X,k)$ & $\Phi^{\mathrm{sym}}(X,k)$\\ \hline \hline

General  & strongly NP-hard & strongly NP-hard  & strongly NP-hard  \\
         & not approximable & not approximable & not approximable \\
         & for $k=2$~\cite{B12}  & for $k=2$~\cite{SAO09} & for unbounded~$k$~\cite{NO13}\\ 
         \cline{2-4}
         &para-NP-hard & para-NP-hard&
         \\
         \cline{2-4}
         &  compact MIP & compact MIP &compact MIP\\ \hline
Nearly acyclic   &   strongly NP-hard &  strongly NP-hard &  \\
 planar$^*$        & not approximable & not approximable & ? \\
         & for unbounded~$k$& for unbounded~$k$& \\ \hline
General acyclic  & $O(|V|^2|A|k^2)$ & $O(|V|^2|A|k^2)$ & $O(|V|^2|A|k^3)$ \\ \hline
Layered  & $O(|A||V|+|V|^2k)$ &  $O(|A||V|+|V|^2k)$ &  $O(|A||V|+|V|^2k)$ \\ \hline
Arc series-parallel   & $O(|A| k^2)$ & $O(|A| k^2)$ & $O(|A| k^2)$   \\ \hline
\multicolumn{4}{l}{$^*$Multidigraph that becomes acyclic planar after removing $O(k)$ arcs.}
 \end{tabular}
 \end{center}
  \end{small}
  \end{table}
 
 \subsection{Recoverable  shortest path problem in general multidigraphs}
\label{srsppm}

In this section we strengthen known hardness results for \textsc{Rec SP} with the neighborhoods $\Phi^{\mathrm{incl}}(X,k)$ and $\Phi^{\mathrm{excl}}(X,k)$
and show compact mixed integer programming formulation (MIP for short) for all three considered neighborhoods.

\subsubsection{Hardness results}

We now analyze the computational complexity of the \textsc{Inc SP} and \textsc{Rec SP} problems
with the neighborhoods~$\Phi^{\mathrm{incl}}(X,k)$, $\Phi^{\mathrm{excl}}(X,k)$
and $\Phi^{\mathrm{sym}}(X,k)$. We refine the known results in this area by providing hardness results for digraphs with a simpler structure, i.e. for the digraphs near acyclic planar ones.
We call them \emph{nearly acyclic planar digraphs}.
A digraph is nearly acyclic planar if
it becomes acyclic planar after removing $O(k)$ arcs.

Consider first the neighborhood~$\Phi^{\mathrm{incl}}(X,k)$. In this case, the \textsc{Inc SP} problem in
a general digraph~$G=(V,A)$ can be solved in $O(k|A|)$ time~\cite{SAO09}, while in~\cite[Corollary~4]{B12} B{\"u}sing proved that the \textsc{Rec SP} problem is strongly 
NP-hard and not approximable in a general digraph unless $\text{P}=\text{NP}$, even
if the recovery parameter~$k=2$.  The proof shown in~\cite{B12} uses
a reduction from the 2-Vertex-Disjoint paths problem, which is strongly NP-complete in general digraphs~\cite{FHW80}.
We now propose a reduction from the $K$-Vertex-Disjoint Paths problem ($K$-\textsc{V-DP} for short) to \textsc{Rec SP},
which will allow us to extend the hardness result to nearly acyclic planar digraphs.

We start by recalling the definition of the $K$-\textsc{V-DP} problem.
We are given a digraph $G=(V,A)$ and $K$ terminal pairs $(s_i,t_i)$, where $s_i,t_i\in V$, $i\in [K]=\{1,\dots,K\}$, are  pairwise distinct vertices.
We ask whether there exist $K$ pairwise vertex-disjoint paths $\pi_1,\ldots,\pi_K$  in~$G$
such that $\pi_i$ is a simple $s_i$-$t_i$ path for each $i\in [K]$.
Naves and Seb\H{o} proved, that if $K$ is a part of the input, then $K$-\textsc{V-DP} is NP-complete even in acyclic planar digraphs \cite{NS09}.

In what follows we  denote by $G+H$
the multidigraph composed of two digraphs~$G=(V, A)$ and~$H=(V_H, A_H)$,
such that $G + H$ has the set of nodes $V$ and the set of arcs containing all the arcs from $A$ and $A_H$.

\begin{lem}\label{lem:kvdp-recsp-reduction}
There exists a polynomial-time reduction from the instance of $K$-\textsc{V-DP} problem
in an acyclic planar digraph $G$ -- constructed based on \textsc{Planar 3-Sat} 
in~\cite[Theorem 4.8]{NS09} -- to \textsc{Rec SP} with $\Phi^{\mathrm{incl}}(X, k)$ 
in a digraph $G + H$. The costs satisfy 
$(C_e, \overline{c}_e) \in \{(0,0), (1,0), (0,1)\}$ for all arcs $e$ in $G + H$, 
where the digraph $H$ is a simple path containing $2k - 1$ arcs 
(i.e. $|A_H| = 2k - 1$). Moreover, the instance of $K$-\textsc{V-DP} is a 
Yes-instance if and only if 
 the total cost of any optimal solution to the
corresponding instance of~\textsc{Rec SP} is zero. 
\label{lcrecsp}
\end{lem}
\begin{proof}
The reduction presented here extends the construction given in~\cite[Theorem 4.8]{NS09}, 
where the \textsc{Planar 3-Sat} problem is reduced to $K$-\textsc{V-DP}. We use 
the same intermediate graph construction and extend it to a graph in \textsc{Rec SP}, thereby 
obtaining a direct reduction from \textsc{Planar 3-Sat} to \textsc{Rec SP}. For 
the sake of completeness, we briefly outline the reduction from 
\textsc{Planar 3-Sat} to $K$-\textsc{V-DP} below.

Let $\phi$ be a Boolean formula in conjunctive normal form, where every clause consists of 
three literals, corresponding to an $(\Xset, \SetOfClauses)$ instance of \textsc{Planar 3-Sat}. 
Here, $\Xset$ is the set of Boolean variables taking values from $\{0, 1\}$, and $\SetOfClauses$ 
is the set of clauses.
We assume each variable~$x \in \Xset$ occurs in exactly three clauses and exactly once as a negative literal
(the problem remains NP-complete \cite{MiddendorfPfeiffer1993}).
We construct a graph $G^{\phi}$ for the instance of $K$-\textsc{V-DP} as follows.
For each clause $C_j \in \SetOfClauses$, the graph $G^{\phi}$ contains a clause gadget 
with three terminal pairs $(s_i, t_i)$ for $i \in [3]$, and three distinguished 
vertices: $v_x^1$, $v_x^2$, and $v_x^3$. For each variable $x_i \in \Xset$, there is 
a variable gadget with a single terminal pair $(s_1, t_1)$ and three distinguished 
vertices: $v_c^1$, $v_c^2$, and $v_c^3$.
We show these gadgets in Figure~\ref{fig:planar-kvdp-gadgets}.
\begin{figure}[htbp]
\begin{center}
\includegraphics[width=0.7\textwidth]{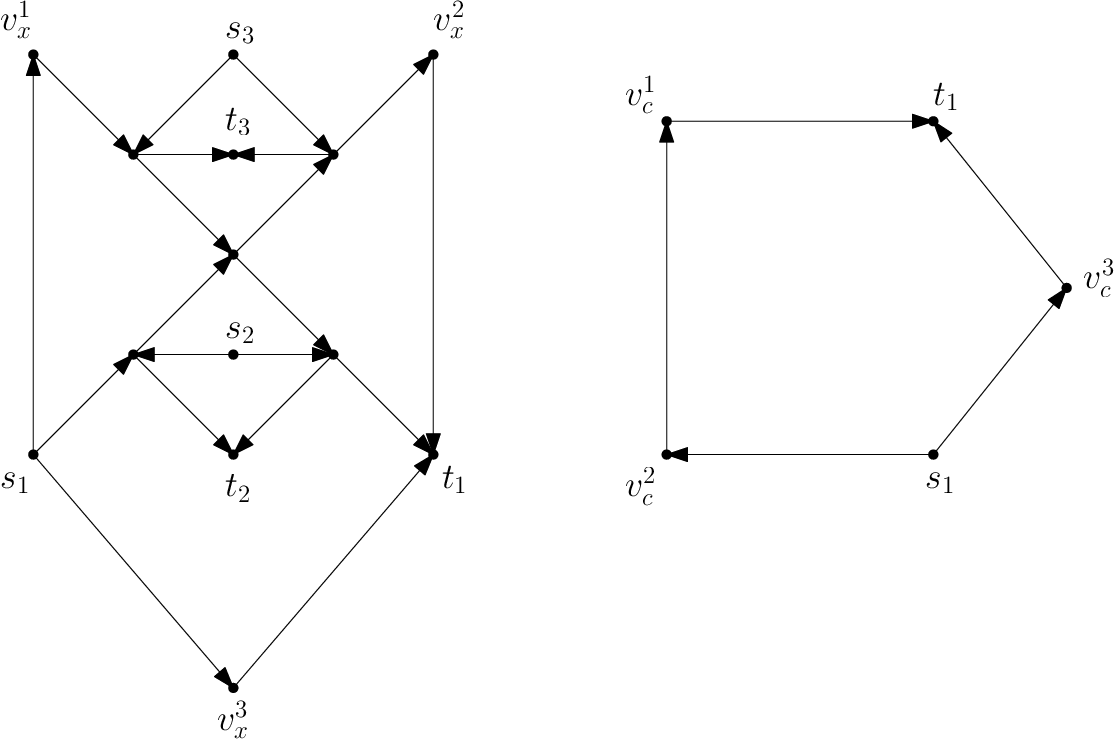}
\caption{
    The clause gadget (left) and variable gadget (right)
    used in the reduction from \textsc{Planar 3-Sat} to $K$-\textsc{V-DP}
    given in~\cite{NS09}.
}
\label{fig:planar-kvdp-gadgets}
\end{center}
\end{figure}
The clause and variable gadgets are connected by identifying their distinguished vertices 
according to the literals occurring in $\phi$,
 i.e.
if the positive literal of variable $x_i$ occurs in clause $C_j$, then either $v_c^1$ 
or $v_c^2$ is identified with one of the distinguished vertices in the clause gadget 
for $C_j$. If the negative literal occurs instead, then $v_c^3$ is identified.
An assignment of~$\Xset$ variables is encoded in the variable gadgets by the~$s_1$-$t_1$ paths.
If the~$s_1$-$t_1$ path in the~$x_i$ variable gadget traverses the~$v_c^{1}$ and~$v_c^{2}$ vertices,
then we have~$x_i = 0$.
Otherwise, it must traverse the~$v_c^{3}$ vertex, and we set~$x_i = 1$.
This finishes the construction of~$G^{\phi}$.
The graph $G^{\phi}$ is planar and  acyclic, which follows directly from the properties 
of the \textsc{Planar 3-Sat} instance and the structure of the gadgets. Furthermore, 
$G^{\phi}$ is acyclic; this is ensured by enforcing an ordering on the variables 
and identifying the vertices of different variable gadgets according to this ordering (see~\cite{NS09} for details).
A set of vertex-disjoint paths in $G^{\phi}$ must be found for all the terminal pairs across all the gadgets.
Thus $K=3|\SetOfClauses|+|\Xset|$. Consequently,
the formula~$\phi$ is satisfiable if and only if~$K$ vertex-disjoint paths exist in~$G^{\phi}$ (\cite[Theorem 4.8]{NS09}).

We now extend the graph $G^{\phi}$ by introducing additional gadgets; for now, 
we consider only their solid arcs, as shown in Figure~\ref{fig:rec-nearly-acyclic-planar-gadgets}. 
The solid arcs in the clause 
gadget (left) and the variable gadget (right) form the  gadgets from the 
aforementioned reduction from \textsc{Planar 3-Sat} to $K$-\textsc{V-DP}.
We also introduce a third type 
of gadget, referred to as the \emph{literal} gadget, which is depicted in the 
middle of Figure~\ref{fig:rec-nearly-acyclic-planar-gadgets} (the solid arcs). Similar to the 
clause gadget, this gadget has three terminal pairs $(s_i, t_i)$ for $i \in [3]$, 
but contains only two distinguished vertices, $l_1$ and $l_2$.
\begin{figure}[h!]
\begin{center}
\includegraphics[width=\textwidth]{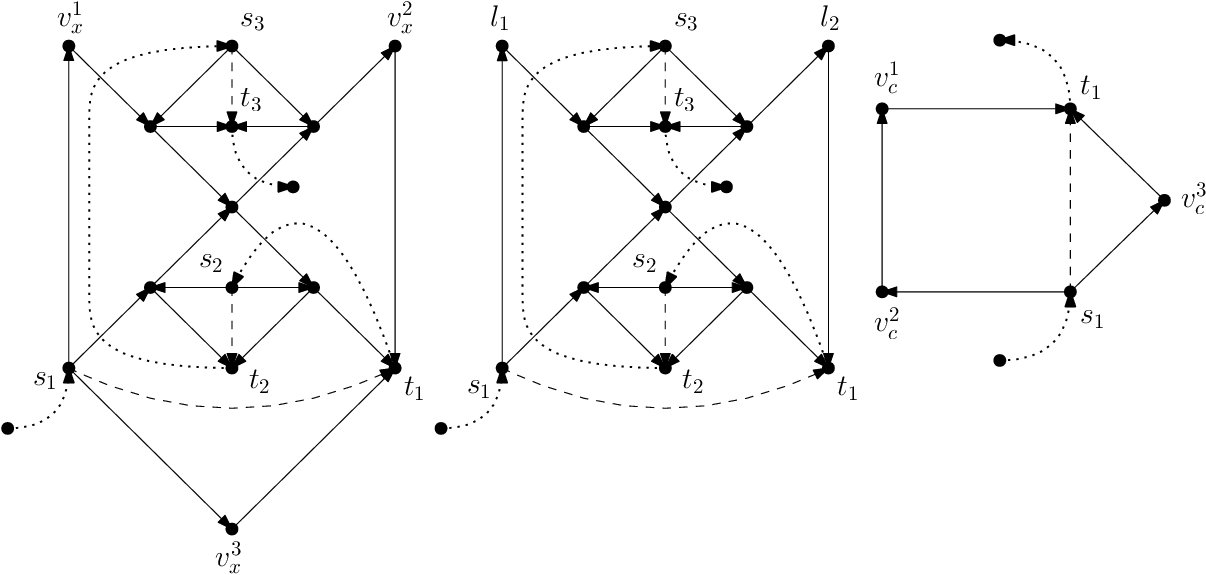}
\caption{
    The clause (left), literal (center) and variable (right) gadgets used in the reduction from \textsc{Planar 3-Sat} to \textsc{Rec SP}.
    The subgraph composed of the solid arcs of the clause and variable gadgets is used 
in the reduction from \textsc{Planar 3-Sat} to $K$-\textsc{V-DP} given in~\cite{NS09}.
}
\label{fig:rec-nearly-acyclic-planar-gadgets}
\end{center}
\end{figure}

The literal gadgets are required to ensure that the paths in a solution to \textsc{Rec SP} 
have a structure implying the existence of $K$ vertex-disjoint paths in $G^{\phi}$, 
which in turn implies a solution to the \textsc{Planar 3-Sat} instance.
These gadgets  are added to $G^{\phi}$ by 
expanding the distinguished vertices $v_x^1$, $v_x^2$, and $v_x^3$ in each clause gadget, 
which were previously identified with the vertices $v_c^1$, $v_c^2$, and $v_c^3$ of the 
corresponding variable gadgets. This expansion consists of identifying the vertex $l_1$ 
of the literal gadget with the distinguished vertex of the clause gadget, and $l_2$ with 
the vertex of the variable gadget. For example, if $v_x^1$ was identified with $v_c^3$ 
in $G^{\phi}$ for a given clause and variable gadget, we expand this connection by 
inserting a new literal gadget in its place, identifying $v_x^1$ with $l_1$ and $v_c^3$ with $l_2$.
Let $G = (V, A)$ denote the graph $G^{\phi}$ after adding the literal gadgets 
using only their solid arcs. 
It is clear that $G$ remains a planar acyclic graph.
Most importantly, an assignment of the variables in $\Xset$ that satisfies $\phi$ 
(if one exists), which is encoded by vertex-disjoint paths in $G^{\phi}$, 
is also encoded in~$G$.
This is based on the following property of the literal gadgets: if the vertex $l_1$ 
is traversed by the $s_1$-$t_1$ path of the  incident clause gadget, then the 
$s_1$-$t_1$ path in the literal gadget can be completed by traversing the vertex $l_2$. 
This, in turn, prevents the $s_1$-$t_1$ path in the incident variable gadget from traversing $l_2$. 
Consequently, $G$ contains an embedded $K$-\textsc{V-DP} instance corresponding 
to the initial \textsc{Planar 3-Sat} formula, where $K = 3|\SetOfClauses| + |\Xset|$. 
At the same time, the graph $G$ constructed in this manner enables the 
extension of the reduction to \textsc{Rec SP}.

We now describe how the dashed and dotted arcs, shown in 
Figure~\ref{fig:rec-nearly-acyclic-planar-gadgets}, are added to $G$; 
these arcs together form a subgraph denoted by $H = (V_H, A_H)$.
The dashed arcs connect all the terminal pairs across all gadgets in $G$. Specifically, 
this set of arcs, denoted by $A_{H_1}$, contains  $(s_i, t_i)$ arc for each $i \in [3]$ 
in every clause and literal gadget, as well as  $(s_1, t_1)$ arc in each variable gadget. 
The set $V_{H_1}$ consists of the terminals incident to these arcs, which together 
form the subgraph $H_1 = (V_{H_1}, A_{H_1})$.
Now we are ready to add the dotted arcs to $G$. For this purpose, 
we assume arbitrary but fixed orderings of the clause, literal, 
and variable gadgets.
We start with the first clause gadget according to the assumed ordering. 
We connect terminal $t_1$ to $s_2$, and $t_2$ to $s_3$, by adding the arcs 
$(t_1, s_2)$ and $(t_2, s_3)$, respectively. Then, we link this gadget 
to the subsequent clause gadget in the ordering by adding an arc from $t_3$ 
to the terminal $s_1$ of that next gadget. This process is repeated for all 
successive clause gadgets. Finally, the last clause gadget is connected to 
the first literal gadget by an arc from $t_3$ of the clause gadget to $s_1$ 
of the literal gadget. 
In a similar manner, within the first literal gadget, we connect terminal $t_1$ 
to $s_2$, and $t_2$ to $s_3$, by adding the arcs $(t_1, s_2)$ and $(t_2, s_3)$, 
respectively. We then link this gadget to the subsequent literal gadget in the 
ordering by adding an arc from $t_3$ to the terminal $s_1$ of that next gadget. 
This process is repeated for all successive literal gadgets. Finally, the last 
literal gadget is connected to the first variable gadget by an arc from $t_3$ 
of the literal gadget to $s_1$ of the variable gadget.
Then, we link this variable gadget to the subsequent variable gadget in the 
ordering by adding an arc from $t_1$ to the terminal $s_1$ of that next gadget. 
The set of these newly added dotted arcs is denoted by $A_{H_2}$, and the set of 
their incident terminals by $V_{H_2}$, which together form the subgraph 
$H_2 = (V_{H_2}, A_{H_2})$. Clearly, $V_{H_1} = V_{H_2}$. 
Finally, we define the source and sink by setting $s = s_1$, where $s_1$ 
is the terminal in the first clause gadget, and $t = t_1$, where $t_1$ is the 
terminal in the last variable gadget. Thus, $H = H_1 + H_2$ forms a simple 
$s$-$t$ path traversing all the clause gadgets, followed by all the literal gadgets, 
and then all the variable gadgets, strictly respecting the fixed orderings.
This completes the construction of the graph $G+H$.
Adding the arcs from $A_H$ introduces cycles to the structure. Moreover, 
the resulting graph $G + H$ is no longer necessarily planar.
Since there are $3|\SetOfClauses|$ literal gadgets in $G$, we have 
$|A_{H_1}| = 12|\SetOfClauses| + |\Xset|$, $|A_{H_2}| = 12|\SetOfClauses| + |\Xset| - 1$, 
and $|A_{H}| = 24|\SetOfClauses| + 2|\Xset| - 1$.

As the last step of the construction,
we assign the first and second-stage costs to the arcs in $G + H$:
\begin{equation}
\label{eq:rec-almost-acyclic-planar-cost-structure-incl}
    (C_e, \overline{c}_e) =
    \begin{cases}
        (0, 1) & \text{if } e \in A, \\
        (1, 0) & \text{if } e \in A_{H_1}, \\
        (0, 0) & \text{if } e \in A_{H_2}.
    \end{cases}
\end{equation}
We also set the recovery parameter to $k = |A_{H_1}| = 12|\SetOfClauses| + |\Xset|$. 
This completes the construction of the instance of \textsc{Rec SP}. 
Since the graph $G$ is planar and acyclic, and $|A_H| = 2k - 1 = O(k)$, 
the combined graph $G + H$ is nearly acyclic planar.

We now prove that there is a satisfying assignment for $\phi$ if and only if 
there exists a solution $(X^*, Y^*)$ to the constructed \textsc{Rec SP} instance 
with a total cost of $\sum_{e \in X^*} C_e + \sum_{e \in Y^*} \overline{c}_e = 0$.

($\implies$) Assume that $\phi$ is satisfiable, and let $\mathbf{x}$ be a 
satisfying assignment. It follows that within the graph $G$, there exist $K = 3|\SetOfClauses| + |\Xset|$
vertex-disjoint paths from $s_i$ to $t_i$ for $i \in [3]$ in the clause gadgets, 
as well as vertex-disjoint $s_1$-$t_1$ paths in the variable gadgets. This 
exploits the fact that $G$ embeds an instance of $K$-\textsc{V-DP} in which these 
$K$ vertex-disjoint paths exist if and only if $\phi$ is satisfiable.
Based on the properties of the literal gadgets, we can now complete the 
vertex-disjoint paths from $s_i$ to $t_i$ for $i \in [3]$ across all the literal 
gadgets. Namely, note that either~$l_1$ or~$l_2$ vertex is 
untraversed by the vertex-disjoint paths in the incident clause and variable 
gadgets, we can then use this free vertex to complete the $s_1$-$t_1$ path. 
Consequently, the remaining two paths, $s_2$-$t_2$ and $s_3$-$t_3$, can always be completed.

Let $P^*$ denote the set of arcs belonging to the vertex-disjoint paths from 
all the clause, literal, and variable gadgets. We construct the first-stage path 
in \textsc{Rec SP} by combining these paths with the arcs that connect 
successive terminal pairs: $X^* = P^* \cup A_{H_2}$. Thanks to the arcs in 
$A_{H_2}$, the path $X^*$ merges all individual vertex-disjoint paths, 
leading  from vertex $s_1$ in the first clause gadget to vertex $t_1$ in the 
last variable gadget. Thus,
it is evident that $X^*$ is a simple $s$-$t$ path. By the cost assignment 
in~\eqref{eq:rec-almost-acyclic-planar-cost-structure-incl}, the first-stage 
cost $\sum_{e \in X^*} C_e$ is equal to~$0$. We define the second-stage path as 
$Y^* = A_H$; it is easy to see that $Y^*$ is also a simple $s$-$t$ path. Since the 
arcs of $A_{H_2}$ belong to both $X^*$ and $Y^*$, it follows that 
$Y^* \in \Phi^{\text{incl}}(X^*, |A_{H_1}|)$. Finally, the second-stage cost 
$\sum_{e \in Y^*} \overline{c}_e$ is also equal to~$0$.

($\impliedby$) 
Assume that there exists a solution $(X^*, Y^*)$ to the \textsc{Rec SP} instance 
in $G + H$ with a total cost of $\sum_{e \in X^*} C_e + \sum_{e \in Y^*} \overline{c}_e = 0$. 
Consequently, both $X^*$ and $Y^*$ must be simple $s$-$t$ paths in $G + H$. 
From the cost structure defined in~\eqref{eq:rec-almost-acyclic-planar-cost-structure-incl}, 
it follows that $Y^*$ is the unique second-stage path with a cost of zero; 
hence, it must be of the form $Y^* = A_H$. Furthermore, the cost structure 
combined with the neighborhood constraint $Y^* \in \Phi^{\text{incl}}(X^*, |A_{H_1}|)$ 
implies that $X^*$ must contain all arcs belonging to $A_{H_2}$. Thus, $X^*$ 
traverses all terminals within the clause, literal, and variable gadgets, 
meaning that it visits all the vertices in $V_{H_2}$. The remaining arcs of $X^*$ 
must come from $A$. Since $X^*$ is a simple $s$-$t$ path, these arcs 
necessarily form vertex-disjoint subpaths starting and ending at the 
respective terminals.
What remains to be proven is that these subpaths of $X^*$ within $G$ are 
pairwise vertex-disjoint. To this, it suffices to show that $X^*$ visits 
the gadget terminals in the correct sequential order—namely, the same order 
in which the path $Y^*$ visits them. That is, after visiting the terminal $s_i$ 
(for $i \in [3]$) in any clause or literal gadget, or $s_1$ in any variable 
gadget, the very next terminal visited by $X^*$ must be the corresponding 
paired terminal $t_i$ (or $t_1$, respectively) within the same gadget. 
It is easy to see that if $X^*$ visits the terminals in this precise sequence, 
its resulting subpaths within~$G$ are necessarily pairwise vertex-disjoint.

We now prove that this must indeed be the case for the considered path $X^*$. 
For the sake of contradiction, let us assume that $X^*$ visits the terminal 
vertices of the gadgets in $G$ in an incorrect order. Since $X^*$ contains 
all arcs of $A_{H_2}$, and the path $Y^*$ traverses all the clause gadgets, 
followed by all the literal gadgets, and finally all the variable gadgets, 
there must exist a gadget in which the path $X^*$ violates this correct 
ordering for the first time.

We need to consider three cases, depending on where this violation occurs. 
We begin with the case where the violation occurs for the first time within 
a clause gadget. In this scenario, $X^*$ leaves the gadget through one of the 
vertices $v_x^{1}$, $v_x^{2}$, or $v_x^{3}$ and subsequently reenters it through 
one of the distinguished vertices to visit the remaining unvisited terminals. 
Note that although the path could also initially leave the gadget through $t_3$, 
it would still eventually be forced to leave it again through one of the 
aforementioned vertices. 
By the construction of $G + H$, the path $X^*$ enters some incident literal 
gadget through $l_1$ and can leave it either through $t_3$ or $l_2$. In both 
cases, at least the terminal $s_1$ within this literal gadget is left unvisited. 
However, when $X^*$ eventually returns to this literal gadget to visit $s_1$, 
the only available entry point is through the terminal $s_1$ itself. Consequently, 
the path cannot subsequently exit the gadget without creating a cycle, thereby 
violating its simplicity. 
More precisely, as illustrated by the routing alternatives comprising the solid 
and dotted arcs shown in Figure~\ref{fig:rec-nearly-acyclic-planar-gadgets}, 
we consider two subcases:
\begin{itemize}
    \item If $X^*$ leaves through $t_3$ during the first visit: Upon 
    returning to visit $s_1$, the path is forced to leave the gadget through 
    $l_2$ during this second visit. As can be seen from the structure of the 
    solid and dotted arcs available to $X^*$, this routing inevitably closes a cycle with the 
    previously traversed segments within the considered  literal gadget.
    \item If $X^*$ leaves through $l_2$ during the first visit: Upon 
    returning to visit $s_1$, the path must leave the gadget through $t_3$ 
    during this second visit. Following the allowed solid and dotted arcs, this movement 
    similarly intersects the history of $X^*$ and results in a cycle, yielding 
    a contradiction in both cases.
\end{itemize}

Consider the case where the violation occurs for the first time within a literal 
gadget. Clearly, the path $X^*$ must have entered this literal gadget through 
$s_1$. It can then leave the gadget through one of three vertices: $l_1$, $l_2$, 
or $t_3$. We analyze these subcases below:
\begin{itemize}
    \item If $X^*$ leaves through $l_1$: The only available next vertex 
    the path can transition to is $t_1$ of the incident clause gadget. However, 
    this immediately creates a cycle because all clause gadgets have already been 
    properly traversed at this point, yielding a contradiction.
    \item If $X^*$ leaves through $l_2$: It becomes impossible for the 
    path to return and visit the remaining unvisited terminals within this 
    gadget. Specifically, any attempt to reenter through $l_1$ would imply that 
    a violation occurs within the incident clause gadget, which directly 
    contradicts the fact that all clause gadgets have already been properly traversed.
    \item If $X^*$ leaves through $t_3$: The path could potentially 
    reenter the gadget through $l_2$. However, if it subsequently attempts to 
    leave through $l_1$, it creates a cycle in the incident clause gadget 
    that has already been properly traversed, which again results in a contradiction.
\end{itemize}
We finally turn to the last case where the violation occurs for the first time 
within a variable gadget. This is only possible if $X^*$ enters this gadget 
through the terminal $s_1$ and leaves it prematurely through one of the vertices 
$v_c^{1}$, $v_c^{2}$, or $v_c^{3}$. The path then enters some incident literal 
gadget through its $l_2$ vertex, from which the only available transition is to 
the terminal $t_1$ of the same literal gadget. However, this immediately 
creates a cycle because all literal gadgets have already been properly traversed 
at this point, yielding a contradiction.

We have thus proved that $X^*$ visits the terminals in the correct sequential 
order -- namely, the same order in which the path $Y^*$ visits them. Consequently, 
its resulting subpaths within $G$ are necessarily pairwise vertex-disjoint. 
Using the fact that our construction extends the reduction from \textsc{Planar 3-Sat} 
to $K$-\textsc{V-DP}, we conclude that there exists a satisfying assignment to~$\phi$.
\end{proof}

We now use the reduction from Lemma \ref{lem:kvdp-recsp-reduction} to prove the hardness of \text{Rec SP} in nearly acyclic planar digraphs.

\begin{thm}
\label{trecplan}
The  \textsc{Rec SP} problem  with~$\Phi^{\mathrm{incl}}(X,k)$
in a digraph $G+H$ with  costs $(C_e,\overline{c}_e)\in 
\{(0,0),(1,0),(0,1)\}$ for all arcs~$e$ in $G+H$
 is strongly NP-hard and not approximable  
unless $\text{P} = \text{NP}$, if  $G$ is a acyclic planar digraph, $k$ is  part of the input, and $H$ is a simple path having $2k-1$ arcs.
\end{thm}
The inapproximability result follows from the fact that any approximation algorithm for \textsc{Rec SP} would detect in polynomial time a solution with the cost equal to~0.
We will show in Section~\ref{srspam} that   \textsc{Rec SP} is polynomially solvable in acyclic digraphs. Theorem~\ref{trecplan} shows that it is enough to add $2k-1$~arcs to an acyclic planar digraph to make \textsc{Rec SP} computationally hard problem.

Let us look at the \textsc{Rec SP} problem from the parameterized complexity point of view, namely
we show that the decision version of  \textsc{Rec SP}, parameterized 
by the recovery parameter~$k$, is para-NP-complete.
The class para-NP is the class of
  all parameterized problems, with some parameter $\kappa$, that can be solved by a nondeterministic algorithm in 
 $f(\kappa)\cdot|\mathcal{I}|^{O(1)}$ time,
 where~$f$ is a computable function and $|\mathcal{I}|$ is the size of the instance~$\mathcal{I}$ of the problem
 (see, e.g.,~\cite{FG06}).
A \emph{slice} of a parameterized problem is its non-parameterized counterpart, obtained by setting~$\kappa$ to a constant value.
If a slice of a parameterized problem is NP-hard, then the problem is para-NP-hard.
We will use this fact in our proofs.
It is known that $\text{FPT} = \text{para-NP}$ if and only if $\text{P} = \text{NP}$, where
 the FPT class consists of all parameterized problems that can be solved in $f(\kappa)\cdot|\mathcal{I}|^{O(1)}$ time, i.e. the problems, which are \emph{fixed-parameter tractable} (see, e.g.,~\cite{FG06} for an in-depth treatment of the parameterized complexity).
Therefore,  para-NP-completeness of a parameterized problem 
  excludes the existence of algorithms for the problem whose running time could contain 
 an exponential term with respect to parameter $\kappa$ only.

\begin{thm} 
\label{thminclp}
The decision version of~\textsc{Rec SP}  with~$\Phi^{\mathrm{incl}}(X,k)$ parameterized by $k$
is para-NP-complete in general digraphs.
\end{thm}
\begin{proof}
The problem~\textsc{Rec SP} is in  NP because \textsc{Inc SP} (\textsc{Adv SP}) with~$\Phi^{\mathrm{incl}}(X,k)$
  can be solved in $O(k|A|)$ time~\cite{SAO09}.
Hence, \textsc{Rec SP} belongs to para-NP.
Since  \textsc{Rec SP} is NP-hard for constant~$k$ ($k=2$)~\cite{B12}, it is para-NP-hard
(see~\cite[Theorem 2.14]{FG06}).
Thus, it is 
para-NP-complete in general digraphs.
\end{proof}
It follows from Theorem~\ref{thminclp} that \textsc{Rec SP}  with~$\Phi^{\mathrm{incl}}(X,k)$ parameterized by $k$
is para-NP-hard in general digraphs.

We now consider the arc exclusion neighborhood~$\Phi^{\mathrm{excl}}(X,k)$.
For this case, the \textsc{Inc SP} problem is already 
 strongly 
NP-hard and not approximable in general digraphs, even if the recovery parameter $k=2$, by a reduction from the $2$-\textsc{V-DP} problem~\cite{SAO09}. 
We extend the hardness result for  \textsc{Inc SP}  to nearly acyclic planar digraphs
\begin{lem}
There exists a polynomial-time reduction from the instance of $K$-\textsc{V-DP} problem
in an acyclic planar digraph $G$ -- constructed based on \textsc{Planar 3-Sat} 
in~\cite[Theorem 4.8]{NS09} --
to  \textsc{Inc SP} with~$\Phi^{\mathrm{excl}}(X,k)$, in a digraph $G+H$
with $\overline{c}_e \in \{0,1\}$ for all arcs~$e$ in $G+H$,
where~$H$ is a simple  path having $2k-1$~arcs ($|A_H|=2k-1$).
Moreover, an instance of $K$-\textsc{V-DP} is  a Yes-instance iff the total cost of any optimal solution to the
corresponding instance of~\textsc{INC SP} is zero.
\label{lredicsex}
\end{lem}
\begin{proof}
We modify the reduction from Lemma \ref{lem:kvdp-recsp-reduction} to work for the arc exclusion neighborhood.
In this case, we invert the roles of the first and second stage paths by 
defining the neighborhood cost structure as follows:
\begin{equation}\label{eq:rec-almost-acyclic-planar-cost-structure-excl}
    (C_e, \overline{c}_e) =
    \begin{cases}
        (1, 0) \text{ if } e \in A,\\
        (0, 1) \text{ if } e \in A_{H_1},\\
        (0, 0) \text{ if } e \in A_{H_2},
    \end{cases}
\end{equation}
where we fix the first-stage solution to $X^* = H$ and set $k = |A_{H_1}|$.
It follows that $Y^*$ must contain~$A_{H_2}$ to satisfy the neighborhood constraint
and complete the subpaths between the terminal vertices using only the arcs from~$G$.
The rest of the proof proceeds the same as in Lemma~\ref{lem:kvdp-recsp-reduction}.
\end{proof}

Using the strong NP-completeness of $K$-\textsc{V-DP} for fixed $K=2$ in general digraphs, proved in~\cite{FHW80},
it has been shown in~\cite{SAO09} that \textsc{Inc SP} with $\Phi^{\mathrm{excl}}(X,k)$ is NP-hard for fixed~$k=2$ in general digraphs.
We now use the reduction from Lemma~\ref{lredicsex} to prove the hardness of  \textsc{Inc SP}  in  simpler digraphs,
when $k$ is part of the input.
\begin{thm}
\label{trecplanex}
The  \textsc{Inc SP} problem  with~$\Phi^{\mathrm{excl}}(X,k)$ in a digraph $G+H$ with  costs $\overline{c}_e\in 
\{0,1\}$ for all arcs~$e$ in $G+H$
 is strongly NP-hard and not approximable  
unless $\text{P} = \text{NP}$, if $G$ is an acyclic planar digraph, $k$ is part of the input and $H$ is a simple path having $2k-1$ arcs.
\end{thm}

We have the parameterized complexity result analogous to Theorem~\ref{thminclp} for the arc exclusion neighborhood.
\begin{thm} 
The decision version of~\textsc{Inc SP}  with~$\Phi^{\mathrm{excl}}(X,k)$ parameterized by $k$
is para-NP-complete in general digraphs.
\end{thm}
\begin{proof}
Clearly~\textsc{Inc SP} is in  NP, and so
 it is in para-NP.
Since~\textsc{Inc SP} is NP-hard for $k=2$~\cite{SAO09}, it is para-NP-hard
(see~\cite[Theorem 2.14]{FG06}).
Thus, it is 
para-NP-complete in general digraphs.
\end{proof}
Finally, consider the arc symmetric difference neighborhood $\Phi^{\mathrm{sym}}(X,k)$.
It was shown in~\cite{NO13} that \textsc{Inc SP} with $\Phi^{\mathrm{sym}}(X,k)$ is strongly NP-hard if the recovery parameter~$k$ is part of the input.
A simple and straightforward modification of the reduction from~\cite{NO13} (it is enough to change some arc costs)
allows us to establish the inapproximability of \textsc{Inc SP} also for this case.
\begin{thm} 
The~\textsc{Inc SP} problem in general digraphs with~$\Phi^{\mathrm{sym}}(X,k)$ is strongly NP-hard and not approximable
unless $\text{P} = \text{NP}$ if $k$ is part of the input.
\end{thm}

Since \textsc{Inc SP} is a special case of~\textsc{Rec SP},
all the hardness results for $\Phi^{\mathrm{excl}}(X,k)$ and $\Phi^{\mathrm{sym}}(X,k)$ remain valid for~\textsc{Rec SP} with these neighborhoods.
Hence, \textsc{Rec SP} is a computationally hard problem for all neighborhoods under consideration.
Notice, however, that the hardness results shown in this section hold for general digraphs.
In Section~\ref{srspam}, we will show that \textsc{Rec SP} can be solved in polynomial time for all the considered neighborhoods if the input graph is an acyclic multidigraph.

\subsubsection{Compact mixed integer programming formulation}
\label{secmiprec}

In this section, we will show that \textsc{Rec SP} in general multidigraphs
can be solved using a compact MIP. Let $\chi (\Phi)\subset \{0,1\}^{|A|}$ be the set of characteristic vectors of simple $s$-$t$ paths in $G=(V,A)$. In the following, we will identify each simple $s$-$t$ path $X$ with its characteristic vector $\pmb{x}\in \chi (\Phi)$.
Likewise, $\chi (\Phi(\pmb{x},k))\subset\{0,1\}^{|A|}$ is the set of characteristic vectors of simple $s$-$t$ paths in the neighborhood of $\pmb{x}\in \chi (\Phi)$. Define $d_i=1$ if $i=s$, $d_i=-1$ if $i=t$ and $d_i=0$ for $i\in V\setminus\{s,t\}$.
 The MIP formulation for \textsc{Rec~SP} takes the following form: 
 \begin{align}
  \min\; & \sum_{e\in A}  C_e x_e+ \sum_{e\in A}  \overline{c}_e y_e \label{cstx0}\\ 
   &\pmb{x}\in \chi (\Phi)\subset\{0,1\}^{|A|},\label{cstx}\\
    & \pmb{y}\in \chi (\Phi(\pmb{x},k))\subset\{0,1\}^{|A|},\label{cneigh}
\end{align}
where $\chi(\Phi)$ is described by the following constraints:
\begin{align}
  & \sum_{(i,j)\in A} x_{i,j}-\sum_{(j,i)\in A} x_{j,i}=d_i & i\in V, \label{chx0} \\
  & x_{i,s}=x_{t,i}=0 & i\in V, \label{chx1}\\
  & p_i+M x_{i,j}+1\leq p_j+M & (i,j)\in A,  \label{chx2}\\
  & p_i\in\{1,\dots,|V|\} & i\in V,  \label{chx3}\\
  & x_{i,j}\in \{0,1\} & (i,j)\in A, \label{chx4}
\end{align}
where $M$ is a sufficiently large constant.
Constraints~(\ref{chx0}) are standard mass-balance constraints~(see, e.g.,~\cite{AMO93}). Constraints (\ref{chx1})--(\ref{chx3}) must be added to ensure that $\pmb{x}$ describes a simple path in~$G$. They form a system of Miller-Tucker-Zemlin type constraints that exclude directed cycles in $\pmb{x}$~(see, e.g.~\cite{BDP16, MTZ60}). Of course, these constraints can be omitted if $G$ is acyclic and $\chi(\Phi)$ is then described only by~(\ref{chx0}) and~(\ref{chx4}). For general digraphs (\ref{chx1})--(\ref{chx3}) must be added, even if all first-stage arc costs are positive. This fact is demonstrated by an example shown in Appendix~\ref{dod}.
The set $\chi (\Phi(\pmb{x},k))$  can be modeled
by the following constraints:
 \begin{align}
          &\chi(\Phi^{\mathrm{incl}}(\pmb{x},k)):  & & \pmb{y}\in \chi (\Phi),\; \sum_{e\in A} (1-x_e)y_e\leq k,\label{cneighincl}\\
    &\chi(\Phi^{\mathrm{excl}}(\pmb{x},k)): & &  \pmb{y}\in \chi (\Phi),\; \sum_{e\in A} (1-y_e)x_e\leq k,\label{cneighexcl} \\
    &\chi(\Phi^{\mathrm{sym}}(\pmb{x},k)): & & \pmb{y}\in \chi (\Phi)\; \sum_{e\in A} ((1-x_e)y_e+(1-y_e)x_e)\leq k.\label{cneighsym}
\end{align}

By the assumption that $\Uset\subset \Rset_{+}^{|A|}$ (the second-stage arc costs are nonnegative), we can drop the anti-cycling constraints in the description of $\chi(\Phi^{\mathrm{incl}}(\pmb{x},k))$. Indeed, a solution $\pmb{y}\in \chi(\Phi^{\mathrm{incl}}(\pmb{x},k))$ contains a simple $s$-$t$ path and possibly some cycles. We can remove all such cycles by fixing some variables $y_e$ to~0, without violating the constraint~(\ref{cneighincl}). However, we cannot do this for the arc exclusion and arc symmetric difference neighborhoods (an example is shown in Appendix~\ref{dod}). Therefore, for general digraphs, the anti-cycling constraints must be present in the description of $\chi(\Phi^{\rm excl}(\pmb{x},k))$ and $\chi(\Phi^{\rm sym}(\pmb{x},k))$.
The constraints in~(\ref{cneighincl})--(\ref{cneighsym}) can be linearized.
 For  example, we can replace~(\ref{cneighincl}) with
 \begin{align}
     &\sum_{e\in A} (y_e-z_e)\leq k &\\
     &z_e\leq x_e &\forall e\in A,\label{cinter1}\\
      &z_e\leq y_e&\forall e\in A,\label{cinter2}\\
      &z_e\geq 0 &\forall e\in A. \label{cinter3}
\end{align}
In much the same way, we can linearize the constraints~(\ref{cneighexcl}) and~(\ref{cneighsym}).
We conclude that
the \textsc{Rec SP} problem with~$\Phi^{\mathrm{incl}}(X,k)$, $\Phi^{\mathrm{excl}}(X,k)$ and $\Phi^{\mathrm{sym}}(X,k)$ 
admits compact MIP formulations.

\subsection{Recoverable  shortest path problem in acyclic multidigraphs}
\label{srspam}
In this section we show that \textsc{Rec SP} with the neighborhoods $\Phi^{\mathrm{incl}}(X,k)$, $\Phi^{\mathrm{excl}}(X,k)$
and $\Phi^{\mathrm{sym}}(X,k)$
can be solved in polynomial time if the input graph $G$ is an acyclic multidigraph. Observe that we can drop the assumption about the non-negativity of arc costs for this class of graphs. We first assume that $G$ is layered, and then we will generalize the result to arbitrary acyclic multidigraphs. 

\subsubsection{Layered multidigraphs}
\label{layermult}

In a layered multidigraph $G=(V,A)$, the set of nodes $V$ is partitioned into $\ell$ disjoint subsets called layers, i.e. $V=
  V_1\cup \cdots \cup V_{\ell}$, and the arcs connect  only the nodes in successive layers, i.e.
  go only
  from nodes in $V_h$ to nodes in $V_{h+1}$ for each
  $h\in [\ell-1]$.
  Without loss of generality, we can assume $s\in V_1$ and $t\in V_{\ell}$. In layered multidigraphs all $s$-$t$ paths have the same cardinality. Therefore, for every path $X\in \Phi$ and any $k=0,\dots,\ell-1$, $\Phi^{\mathrm{incl}}(X,k)=\Phi^{\mathrm{excl}}(X,k)$. Also $\Phi^{\mathrm{sym}}(X,k)=\Phi^{\mathrm{incl}}(X,\lfloor k/2 \rfloor)$. Therefore, it is enough to construct an algorithm for the arc inclusion neighborhood $\Phi^{\mathrm{incl}}(X,k)$.
  
  In the following, we will show a polynomial transformation from \textsc{Rec SP} in a layered multidigraph to 
  the Constrained Shortest Path Problem (\textsc{CSP} for short) in an acyclic multidigraph. In the \textsc{CSP} problem, 
  we are given an acyclic multidigraph~$G=(V,A)$  with a source~$s\in V $ and a  sink~$t\in V$.
  Each arc $e\in A$ has a cost~$c_e$ and a transition time~$t_e$.  A positive total transition time limit~$T$ is specified. 
  We seek an $s$-$t$ path $\pi$ in $G$ that minimizes the
total cost, subject to not exceeding the transition time limit~$T$.
The  \textsc{CSP} is known to be weakly NP-hard (see, e.g.,~\cite{GJ79}) and can be solved in $O(|A|T)$ time, which is pseudopolynomial~\cite{H92}.

\begin{lem}
\label{lem1}
There is a polynomial time reduction from  the \textsc{Rec SP} problem with the neighborhoods $\Phi^{\mathrm{incl}}(X,k)$,
$\Phi^{\mathrm{excl}}(X,k)$
and $\Phi^{\mathrm{sym}}(X,k)$
 in  
 a layered multidigraph $G=(V,A)$ to the  \textsc{CSP} problem in an acyclic multidigraph with $T=k$.
\end{lem}
\begin{proof}
Let $G=(V,A)$ be a layered multidigraph in an instance of \textsc{Rec SP} with a recovery parameter $k$. We will show a construction for the arc inclusion neighborhood $\Phi^{\mathrm{incl}}(X,k)$.
 We build an acyclic multidigraph $G'=(V',A')$ in the corresponding instance of CSP as follows.
 We fix $V'=V$, $s'=s$, $t'=t$, and  $A'$ contains two types of arcs labelled as $e^{(0)}$ and $e^{(1)}$. 
 Namely,
 for each pair of nodes $i\in V_{h}$ and $j\in V_{h+1}$, $h\in [\ell-1]$, such that there exists at least one arc 
 from~$i$ to~$j$ (note that $G$ can be a multidigraph), 
 we add to $A'$ arc $e^{(0)}=(i,j)$ with the cost $c_{e^{(0)}}=\min_{e=(u,v)\in A\,:\,u=i, v=j} \{C_e+\overline{c}_e\}$
 and the transition time~$t_{e^{(0)}}=0$. 
 For each pair of nodes $i\in V_g$ and $j\in V_h$, $g,h\in [\ell]$, such that the node~$j$ is reachable from~$i$ and
 $1\leq h-g\leq k$,
 we add to $A'$ arc $e^{(1)}=(i,j)$ with the cost $c_{e^{(1)}}$ being the sum
 of the cost of a shortest path~$X^*_{ij}$ from~$i$ to~$j$ under~$C_e$, $e\in A$, and 
 the cost of a shortest path~$Y^*_{ij}$ from~$i$ to~$j$ under~$\overline{c}_e$, $e\in A$, i.e.
  $c_{e^{(1)}}=\sum_{e\in X^*_{ij}}+\sum_{e\in Y^*_{ij}}\overline{c}_e$.
  The transition time is~$t_{e^{(1)}}=|Y^*_{ij}\setminus X^*_{ij}|$.
  Since $G$ is layered,
 $|X^*_{ij}|=|Y^*_{ij}|= h-g$ and thus
 $|Y^*_{ij}\setminus X^*_{ij}| \leq  h-g$.
 Finally, we set $T=k$. The graph $G'$ can be constructed in $O(|A||V|)$ time. Indeed, given a node $i\in V$, we can compute in $O(|A|)$ time the paths $X^*_{ij}$ and $Y^*_{ij}$, for all $j\in V$ reachable from $i$, using a dynamic programming algorithm (see, e.g.~\cite{AMO93}).
 This algorithm must be executed $O(|V|)$ times, which yields the overall running time $O(|A||V|)$.

 We need to show that there is a pair of  paths $X\in \Phi$ and $Y\in \Phi^{\mathrm{incl}}(X,k)$ 
 in~$G$,
 feasible to~\textsc{Rec SP}, with the cost  $\sum_{e\in X} C_e+\sum_{e\in Y}\overline{c}_e\leq UB$ 
 if and only if  there is a path~$\pi$ in~$G'$,  feasible to~\textsc{CSP},
  such that $\sum_{e\in \pi} c_e\leq UB$, where $UB\in \Rset$.
  
($\Rightarrow$) Let $X\in \Phi$ and $Y\in \Phi^{\mathrm{incl}}(X,k)$  
be a pair of paths, feasible to~\textsc{Rec SP} in~$G$,
  with the cost  $\sum_{e\in X} C_e+\sum_{e\in Y}\overline{c}_e\leq UB$.
  We form the corresponding $s$-$t$ path~$\pi$ in~$G'$ as follows.
  Let $(i_1,i_2,\dots,i_l)$, where $i_1=s$ and $i_l=t$ be the sequence of the common nodes of paths $X$ and $Y$ in the order of visiting. This sequence is the same for $X$ and $Y$ because $G$ is acyclic (an example is shown in Figure~\ref{figproof1}). For each $i=2,\dots, l$, if $e=(i_{j-1},i_j)\in X\cap Y$, then we add to~$\pi$  the arc  $e^{(0)}=(i_j,i_{j-1})\in A'$. Notice that the cost of $e^{(0)}$ is not greater than $C_e+\overline{c}_e$ and its transition time is~0.  If $e=(i_{j-1},i_j)\notin X\cap Y$, then there are two disjoint subpaths $X'$ of $X$ and $Y'$ of $Y$ from $i_{j-1}$ to $i_j$ in $G$.
  In this case, we add to~$\pi$  the arc  $e^{(1)}=(i_j,i_{j-1})\in A'$. Notice that the cost of $e^{(1)}$ is not greater than $\sum_{e\in X'} C_e+ \sum_{e\in Y'} \overline{c}_e$, because it is the sum of the costs of the shortest $i_{j-1}-i_j$ paths with respect to $C_e$ and $\overline{c}_e$.
  Also, the transition time of $e^{(1)}$ is not greater than $|Y'\setminus X'|$,
  which follows from the layered structure of the input graph~$G$ (notice that $|Y'\setminus X'|$ is equal to the number of layers between $i_{j-1}$ and $i_j$). The arc $e^{(0)}$ (resp. $e^{(1)}$) must exist in $G'$, because $i_j$ is reachable from $i_{j-1}$ in~$G$. Therefore, $\pi$ is an $s$-$t$ path in $G'$. It is also feasible because its total transition time is at most $|Y\setminus X|\leq k$.
   Finally, the total cost of $\pi$ is not greater than $\sum_{e\in X} C_e+\sum_{e\in Y}\overline{c}_e\leq UB$.
    \begin{figure}[ht]
    \centering
    \includegraphics[height=5cm,width=14cm]{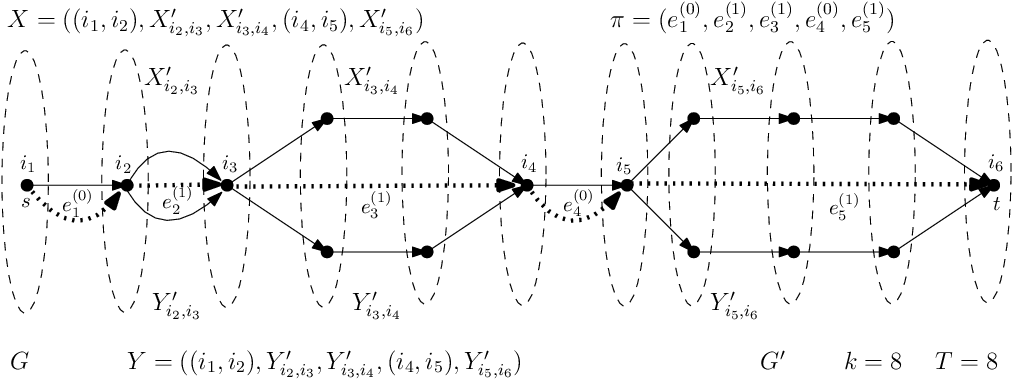}
  \caption{Illustration of the proof of Lemma~\ref{lem1}
  for $\Phi^{\mathrm{incl}}(X,k)$ and $\Phi^{\mathrm{excl}}(X,k)$. 
  The dashed circles denote the layers. The pair of paths $X$ and $Y$ in $G$ corresponds to the path 
  $\pi$ in $G'$. } \label{figproof1}
  \end{figure}
  
($\Leftarrow$) Assume that $\pi$ is an $s$-$t$ path in~$G'$, feasible to~\textsc{CSP}, with
  $\sum_{e\in \pi} c_e\leq UB$.
  For each arc $e^{(0)}=(i,j)\in \pi$, we add the arc $e=(i,j)\in A$ such that 
  $C_e+\overline{c}_e=c_{e^{(0)}}$ to~$X$ and to~$Y$.
  For each arc $e^{(1)}=(i,j)\in \pi$, we add the arcs corresponding to~$X^*_{ij}$ to~$X$ and the
  arcs corresponding to~$Y^*_{ij}$ to~$Y$. It is easily seen that
  $X\in \Phi$,  $Y\in \Phi^{\mathrm{incl}}(X,k)$  
   and 
   $UB\geq \sum_{e\in \pi} c_e= \sum_{e\in X} C_e+\sum_{e\in Y}\overline{c}_e$.
    \end{proof}
    We get the following result:
\begin{thm}
\textsc{Rec SP} with the neighborhoods $\Phi^{\mathrm{incl}}(X,k)$,
$\Phi^{\mathrm{excl}}(X,k)$
and $\Phi^{\mathrm{sym}}(X,k)$
 in  
 a layered multidigraph~$G=(V,A)$ can be solved in $O(|A||V|+|V|^2k)$ time.
\end{thm}
\begin{proof}
The graph $G'$ from the proof of Lemma~\ref{lem1} can be constructed in $O(|A||V|)$ time. The number of arcs in $G'$ is $O(|V|^2)$, so the corresponding CSP problem can be solved in $O(|V|^2 k)$ time, where $k\leq |A|$. Hence the overall running time is $O(|A||V|+|V|^2k)$.
\end{proof}

\subsubsection{General acyclic multidigraphs}

If the input graph $G$ is layered, then the disjoint subpaths of $X$ and $Y$ between any pair of nodes $i$ and $j$ have the same cardinalities.  This fact has been exploited in Section~\ref{layermult}.
However, in general acyclic multidigraphs the paths from $i$ to $j$ can have different cardinalities, so cheaper recoverable action can require choosing more arcs. To overcome this problem, we will extend the construction of the graph $G'$
shown in the proof of Lemma~\ref{lem1}.  However, the neighborhoods are not equivalent now, and we must consider them separately.

\begin{lem}
\label{lem2}
There is a polynomial  time reduction from the \textsc{Rec SP} problem with the neighborhood  $\Phi^{\mathrm{incl}}(X,k)$
 in an acyclic multidigraph $G=(V,A)$ to the  \textsc{CSP}  problem in an acyclic multidigraph with $T=k$.
\end{lem}
\begin{proof}
Given an  acyclic multidigraph $G=(V,A)$ and a recovery parameter $k$ in an instance of~\textsc{Rec SP},
we construct an acyclic multidigraph~$G'=(V',A')$ in the corresponding instance of~\textsc{CSP}.
We set $V'=V$, $s'=s$, $t'=t$, and the set of arcs~$A'$ can  contain~$k+1$ types of arcs labeled as  $e^{(0)},\dots,e^{(k)}$, respectively.
 For each pair of nodes $i,j\in V$, $i\not= j$, such that there exists at least one arc 
 from~$i$ to~$j$, 
 we add to $A'$ the arc $e^{(0)}=(i,j)$ with the cost $c_{e^{(0)}}=\min_{e=(k,l)\in A\,:\,k=i, l=j} \{C_e+\overline{c}_e\}$
 and the transition time~$t_{e^{(0)}}=0$. 
 Consider a pair of nodes $i,j\in V$, $i\not=j$, such that $j$ is reachable from $i$. We first compute a path from $i$ to $j$ having the minimum number of arcs $L^*_{ij}$. If $L_{ij}^*>k$, then we do nothing. 
 Otherwise, we find a shortest path $X^*_{ij}$
 from $i$ to $j$ with the costs $C_e$.
  Then, for each $l=L_{ij}^*,\dots,k$, 
  we find a shortest path $Y^{*(l)}_{ij}$
   with the costs $\overline{c}_e$, $e\in A$, subject to the condition that $Y^{*(l)}_{ij}$ 
   has at most $l$ arcs. Observe that $Y^{*(l)}_{ij}$ 
    can be found in $O(|A|k)$ time by solving a CSP problem with arc transition times equal to~1 and 
 the transition time limit equal to~$l$.
 We add to $A'$ the arc $e^{(l)}=(i,j)$ with the cost 
 $c_{e^{(l)}}=\sum_{e\in X^*_{ij}} C_e+
 \sum_{e\in Y^{*(l)}_{ij}} \overline{c}_e$
 and with the transition time~$t_{e^{(l)}}=l$.
  Finally, we set $T=k$. Graph $G'$ can be constructed in $O(|V|^2 |A| k^2)$ time since we have to solve at most $k$ CSP problems for at most $|V|^2$ pairs of nodes.
As in Lemma~\ref{lem1}, we need to show that  there is a pair of  paths $X\in \Phi$ and 
$Y\in \Phi^{\mathrm{incl}}(X,k)$
 in~$G$,
 feasible to~\textsc{Rec SP}, with the cost  $\sum_{e\in X} C_e+\sum_{e\in Y}\overline{c}_e\leq UB$ 
 if and only if  there is a path~$\pi$ in~$G'$,  feasible to~\textsc{CSP},
  such that $\sum_{e\in \pi} c_e\leq UB$, where $UB\in \Rset$.

($\Rightarrow$) 
The proof goes the same way as that of Lemma~\ref{lem1}. Now, if $e=(i_{j-1},i_j)\notin X\cap Y$, then there are two disjoint subpaths $X'$ and $Y'$ from $i_{j-1}$ to $i_j$ in $G$.
The subpath $Y'$ has $l\in \{L^*_{i_{j-1}i_j},\dots,k\}$ arcs. We add to~$\pi$ the arc $e^{(l)}=(i_{j-1},i_j)\in A'$. 
The transition time of $e^{(l)}$ is $|Y'\setminus X'|=l$ 
and its cost is at most $\sum_{e\in X'} C_e+\sum_{e\in Y'} \overline{c}_e$. The rest of the proof is the same as in Lemma~\ref{lem1}  (an example is depicted in Figure~\ref{figproof2}).
 \begin{figure}[ht]
    \centering
    \includegraphics[height=5cm,width=14cm]{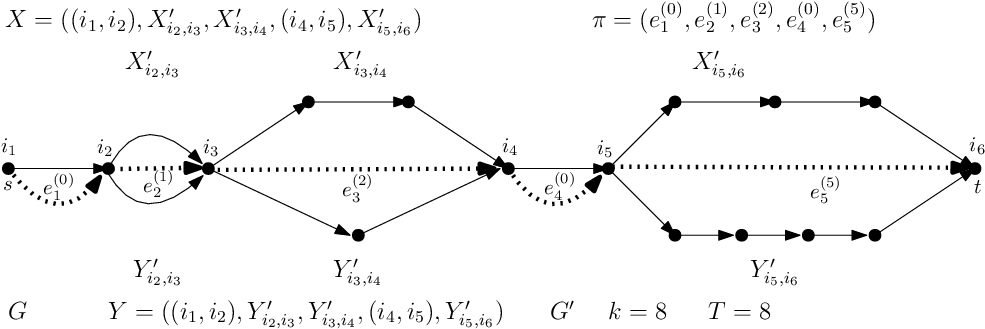}
  \caption{Illustration of the proof of Lemma~\ref{lem2}.  The pair of paths $X$ and $Y$ for $\Phi^{\mathrm{incl}}(X,k)$ 
  in $G$ corresponds to the path $\pi$ in~$G'$.} \label{figproof2}
  \end{figure}

 ($\Leftarrow$) Assume that $\pi$ is an $s$-$t$ path in~$G'$ feasible to~\textsc{CSP} with
  $\sum_{e\in \pi} c_e\leq UB$.
  For each arc $e^{(0)}=(i,j)\in \pi$, we add the arc $e=(i,j)\in A$ such that 
  $C_e+\overline{c}_e=c_{e^{(0)}}$ to~$X$ and to~$Y$.
  For each arc $e^{(l)}=(i,j)\in \pi$, we add the arcs corresponding to~$X^*_{ij}$
   to~$X$
  and the
  arcs corresponding to~$Y^{*(l)}_{ij}$
   to~$Y$.
   We recall that~$|Y^{*(l)}_{ij}|\leq l$.
  It follows easily  that
  $X\in \Phi$, $Y\in \Phi^{\mathrm{incl}}(X,k)$
  and 
   $UB\geq \sum_{e\in \pi} c_e= \sum_{e\in X} C_e+\sum_{e\in Y}\overline{c}_e$.
\end{proof}

\begin{lem}
\label{lem2a}
There is a polynomial  time reduction from the  \textsc{Rec SP} problem with the neighborhood $\Phi^{\mathrm{excl}}(X,k)$
 in an acyclic multidigraph $G=(V,A)$ to the  \textsc{CSP}  problem in an acyclic multidigraph with $T=k$.
\end{lem}
\begin{proof}
The reduction is similar to that in the proof of Lemma~\ref{lem2}. We only point out some differences. Having computed $L^*_{ij}\leq k$, we find a shortest path $Y^*_{ij}$ from $i$ to $j$ with the arc costs $\overline{c}_e$. Then, for each $l=L_{ij}^*,\dots,k$, we find a shortest path $X_{ij}^{*(l)}$ with the arc costs $C_e$, subject to the condition that $X_{ij}^{*(l)}$ has at most $l$ arcs. We add to $A'$ the arc $e^{(l)}=(i,j)$ with the cost 
 $c_{e^{(l)}}=\sum_{e\in Y^*_{ij}} \overline{c}_e+
 \sum_{e\in X^{*(l)}_{ij}} C_e$
 and with the transition time~$t_{e^{(l)}}=l$. The proof goes then similarly to the proof of Lemma~\ref{lem2} with the following modifications:

($\Rightarrow$) 
If $e=(i_{j-1},i_j)\notin X\cap Y$, then 
the subpath $X'$ from $i_{j-1}$ to $i_j$ in $G$ has $l\in \{L^*_{i_{j-1}i_j},\dots,k\}$ arcs. We add to~$\pi$ the arc $e^{(l)}=(i_{j-1},i_j)\in A'$.  The transition time of $e^{(l)}$ is $|X'\setminus Y'|=l$ 
and its cost is at most $\sum_{e\in X'} C_e+\sum_{e\in Y'} \overline{c}_e$.

 ($\Leftarrow$) 
  For each arc $e^{(l)}=(i,j)\in \pi$, we add the arcs corresponding to~$Y^*_{ij}$
   to~$Y$
  and the
  arcs corresponding to~$X^{*(l)}_{ij}$
   to~$X$ and
  we note that~$|X^{*(l)}_{ij}|\leq l$.
  It follows then easily  that
  $X\in \Phi$, $Y\in \Phi^{\mathrm{excl}}(X,k)$ and 
   $UB\geq \sum_{e\in \pi} c_e= \sum_{e\in X} C_e+\sum_{e\in Y}\overline{c}_e$.
\end{proof}

\begin{lem}
\label{lem2b}
There is a polynomial  time reduction from the  \textsc{Rec SP} problem with the neighborhood $\Phi^{\mathrm{sym}}(X,k)$
 in an acyclic multidigraph $G=(V,A)$ to the  \textsc{CSP}  problem in an acyclic multidigraph with $T=k$.
\end{lem}
\begin{proof}
 We again point out some differences in the reduction from Lemma~\ref{lem2}.  Having computed $L_{ij}^*$, if $2L_{ij}^*>k$, then we do nothing. 
  Otherwise, 
  for each $u=L_{ij}^*,\dots,k$ and $v=L_{ij}^*,\dots,k$,
  we find a shortest path $X^{*(u)}_{ij}$ 
   with the costs $C_e$, $e\in A$, subject to
    the constraint $|X^{*(u)}_{ij}|\leq u$ and
     a shortest path $Y^{*(v)}_{ij}$ 
   with the costs $\overline{c}_e$, $e\in A$, subject to
    the constraint $|Y^{*(v)}_{ij}|\leq v$. We
     add to $A'$ the arc $e^{(uv)}=(i,j)$ with the cost 
 $c_{e^{(uv)}}=\sum_{e\in  X^{*(u)}_{ij}} C_e+
 \sum_{e\in Y^{*(v)}_{ij}} \overline{c}_e$
 and the transition time~$t_{e^{(uv)}}=u+v$.  Graph $G'$ can be constructed in $O(|V|^2 |A| k^3)$ time because we have to solve at most $k^2$ \textsc{CSP} problems for at most $|V|^2$ pairs of nodes.
 The proof goes then similarly to the proof of Lemma~\ref{lem2} with the following modifications:
   
  ($\Rightarrow$) 
   If $e=(i_{j-1},i_j)\notin X\cap Y$, then 
the subpaths $X'$ and~$Y'$ from $i_{j-1}$ to $i_j$ in $G$ have $u,v\in \{L^*_{i_{j-1}i_j},\dots,k\}$ arcs, respectively.
We add to~$\pi$ the arc $e^{(uv)}=(i_{j-1},i_j)\in A'$. 
The transition time of $e^{(uv)}$ is $|(Y'\setminus X')\cup(X'\setminus Y') |=u+v$ 
and its cost is at most $\sum_{e\in X'} C_e+\sum_{e\in Y'} \overline{c}_e$. 
 
 ($\Leftarrow$) 
  For each arc $e^{(uv)}=(i,j)\in \pi$, we add the arcs corresponding to~$X^{*(u)}_{ij}$ to~$X$  and the
  arcs corresponding to~$Y^{*(v)}_{ij}$
    to~$Y$, where
  $|X^{*(u)}_{ij}|\leq u$ and $|Y^{*(v)}_{ij}|\leq v$. Thus,
  $X\in \Phi$, $Y\in \Phi^{\mathrm{sym}}(X,k)$   and 
   $UB\geq \sum_{e\in \pi} c_e= \sum_{e\in X} C_e+\sum_{e\in Y}\overline{c}_e$.    
\end{proof}

\begin{thm}
\textsc{Rec SP} in acyclic multidigraph $G=(V,A)$ can be solved in  $O(|V|^2 |A| k^2)$ time for the neighborhoods~$\Phi^{\mathrm{incl}}(X,k)$, $\Phi^{\mathrm{excl}}(X,k)$  and in  $O(|V|^2 |A| k^3)$ time for the neighborhood $\Phi^{\mathrm{sym}}(X,k)$.
\end{thm}
\begin{proof}
  The multidigraphs $G'$ from Lemma~\ref{lem2} and Lemma~\ref{lem2a} can be constructed in $O(|V|^2|A|k^2)$
  time and have at most $|V|^2(k+1)$ arcs. The \textsc{CSP} problem in $G'$ can then be solved in  $O(|V|^2k^2)$ time. The overall running time for $\Phi^{\mathrm{incl}}(X,k)$ and $\Phi^{\mathrm{excl}}(X,k)$ is $O(|V|^2|A|k^2)$.
  Accordingly, the multidigraph $G'$ from Lemma~\ref{lem2b} can be constructed in  $O(|V|^2|A|k^3)$ time and has at most $|V|^2(k^2+1)$ arcs. The \textsc{CSP} problem in $G'$ can then be solved in  $O(|V|^2k^3)$ time. Hence, the overall running time for $\Phi^{\mathrm{sym}}(X,k)$ is $O(|V|^2|A|k^3)$.
\end{proof}

\subsubsection{Arc series-parallel multidigraph}

In this section, we will construct algorithms with better running time for \textsc{Rec SP} with all the considered neighborhoods, assuming that $G$ is an arc series-parallel multidigraph. Again, we have to treat all the considered neighborhoods separately, as the subpaths between two nodes of $G$ may have different cardinalities. 
Let us recall that an arc series-parallel multidigraph (ASP) is
recursively defined as follows (see, e.g.,~\cite{VTL82}). A graph
consisting of two nodes joined by a single arc is ASP. If $G_1$
and $G_2$ are ASP, so are the multidigraphs constructed by each of
the following two operations:
\begin{itemize}
\item \emph{parallel composition} $p(G_1,G_2)$: identify the
source of $G_1$ with the source of $G_2$ and the sink of $G_1$
with the sink of~$G_2$,
\item \emph{series composition}
$s(G_1,G_2)$: identify the sink of $G_1$ with the source of~$G_2$.
\end{itemize}
Each ASP multidigraph~$G$  is associated with a rooted binary
tree~$\mathcal{T}$, called \emph{the binary decomposition tree
of~$G$}, which can be constructed in $O(|A|)$ time~\cite{VTL82}.
 Each leaf of the tree represents an arc in~$G$. Each
internal node $\sigma$ of~$\mathcal{T}$ is labeled  by $\mathtt{S}$
or $\mathtt{P}$ and corresponds to the series or parallel composition
in $G$. Every node $\sigma$ of $\mathcal{T}$ corresponds to  an
ASP subgraph of~$G$, denoted by $G_{\sigma}$, defined by the subtree rooted at~$\sigma$.
The root of~$\mathcal{T}$ represents the input ASP multidigraph~$G$.

In~\cite{B12}, B{\"u}sing showed an idea of a polynomial-time algorithm for 
\textsc{Rec SP} with~$\Phi^{\mathrm{incl}}(X,k)$ in 
 ASP multidigraphs, together with a theorem resulting from it~\cite[Theorem~5]{B12}.
 In this section, we will describe a complete $O(|A|k^2)$-time algorithm for \textsc{Rec SP} in ASP multidigraphs with the  neighborhoods~$\Phi^{\mathrm{incl}}(X,k)$, $\Phi^{\mathrm{excl}}(X,k)$
and $\Phi^{\mathrm{sym}}(X,k)$.
Let $\Phi_{G_{\sigma}}$ denote the set of all paths from the source to the sink in~$G_{\sigma}$ and
\begin{align*}
 & \Phi^{l}_{G_{\sigma}}=\{X\in \Phi_ {G_{\sigma}}\,:\,|X|=l\},\\
 & \Psi^{\mathrm{incl}[l]}_{G_{\sigma}}=\{ (X,Y)  \in\Phi_{G_{\sigma}}\times \Phi_{G_{\sigma}} \,: |Y\setminus X|= l\},\\& \Psi^{\mathrm{excl}[l]}_{G_{\sigma}}=\{ (X,Y)  \in\Phi_{G_{\sigma}}\times \Phi_{G_{\sigma}} \,: |X\setminus Y|= l\}, \\& \Psi^{\mathrm{sym}[l]}_{G_{\sigma}}=\{ (X,Y)  \in\Phi_{G_{\sigma}}\times \Phi_{G_{\sigma}} \,: |Y\setminus X|+|X\setminus Y|= l\}.
 \end{align*}
  
Consider first the arc inclusion neighborhood $\Phi^{\mathrm{incl}}(X,k)$.
For each ASP subgraph~$G_{\sigma}$, we store the following three crucial pieces of information. 
The first one is the cost of a shortest path from  the source  to  the sink in~$G_{\sigma}$  under costs~$C_e$, $e\in A$:
\begin{equation}
C_{G_{\sigma}}=\min_{X\in \Phi_{G_{\sigma}}} \sum_{e\in X} C_e.
\label{locost}
\end{equation}
The second one is the $k$-element array~$\overline{\pmb{c}}_{G_{\sigma}}$, whose $l$-th element~$\overline{\pmb{c}}_{G_{\sigma}}[l]$
is the cost of a shortest path $Y^*\in \Phi^{l}_{G_{\sigma}}$  under the costs~$\overline{c}_e$, $e\in A$, if such a path exists:
\begin{equation}
\overline{\pmb{c}}_{G_{\sigma}}[l]=
\begin{cases}
\displaystyle \min_{Y\in \Phi^{l}_{G_{\sigma}}} \sum_{e\in Y}\overline{c}_e &\text{if 
$\Phi^{l}_{G_{\sigma}}\not=\emptyset$,}\\
+\infty&\text{otherwise},
\end{cases}
\quad l=1,\ldots,k.
\label{upcost}
\end{equation}
The third one is the $(k+1)$-element array~$\pmb{c}^*_{G_{\sigma}}$, whose $l$-th element~$\pmb{c}^*_{G_{\sigma}}[l]$,
stores the cost of an optimal pair $(X^*,Y^*) \in\Psi^{\mathrm{incl}[l]}_{G_{\sigma}}$,
if such  paths exist:
\begin{equation}
\label{defcs}
\pmb{c}^*_{G_{\sigma}}[l]=
\begin{cases}
\displaystyle\min_{(X,Y)\in  \Psi^{\mathrm{incl}[l]}_{G_{\sigma}}} \left(\sum_{e\in X} C_e +  \sum_{e\in Y} \overline{c}_e\right)&
\text{if $ \Psi^{\mathrm{incl}[l]}_{G_{\sigma}}\not=\emptyset $},\\
+\infty&\text{otherwise},
\end{cases}
\quad l=0,\ldots,k.
\end{equation}
Note that at least one set of paths $\Psi^{\mathrm{incl}[l]}_{G_{\sigma}}$ is always
nonempty in~$G_{\sigma}$, for instance, when $l=0$.
For each leaf node~$\sigma$ of the tree~$\mathcal{T}$ (in this case, the  subgraph~$G_{\sigma}$ consists of
a single arc~$e\in A$), the initial values of $C_{G_\sigma}$, $\overline{\pmb{c}}_{G_{\sigma}}$ and $\pmb{c}^*_{G_{\sigma}}$ are as follows:
\begin{align}
&C_{G_{\sigma}}= C_e,& \label{ini1}\\
&\overline{\pmb{c}}_{G_{\sigma}}[1]=\overline{c}_e,\;\overline{\pmb{c}}_{G_{\sigma}}[l]=+\infty, &l=2,\ldots,k,  \label{ini2}\\
&\pmb{c}^*_{G_{\sigma}}[0]=C_e+\overline{c}_e,\;\pmb{c}^*_{G_{\sigma}}[l]=+\infty, &l=1,\ldots,k.  \label{ini3}
\end{align}
By traversing the tree~$\mathcal{T}$ from the leaves to the root, we recursively 
construct each  ASP subgraph~$G_{\sigma}$ corresponding to the internal node~$\sigma$ and
compute 
the values of $C_{G_\sigma}$, $\overline{\pmb{c}}_{G_{\sigma}}$ and $\pmb{c}^*_{G_{\sigma}}$ associated with~$G_{\sigma}$,
 depending on the label of the node~$\sigma$.
If $\sigma$ is marked~$\mathtt{P}$, then 
we call Algorithm~\ref{parallel},
otherwise we  call Algorithm~\ref{series} on the subgraphs~$G_{\mathrm{left}(\sigma)}$ and $G_{\mathrm{right}(\sigma)}$,
where $\mathrm{left}(\sigma)$ and $\mathrm{right}(\sigma)$ are children of~$\sigma$ in~$\mathcal{T}$.
If $\sigma$ is the root of~$\mathcal{T}$, then  $G_{\sigma}=G$ and the array~$\pmb{c}^*_{G_{\sigma}}$
contains information about the cost 
of an optimal 
pair of paths~$X^*,Y^*\in \Phi$, 
being an optimal solution to \textsc{Rec SP} in~$G$, which is equal to 
\[
\min_{0\leq l\leq k} \pmb{c}^*_{G_{\sigma}}[l].
\]

The algorithm for \textsc{Rec SP} with~$\Phi^{\mathrm{incl}}(X,k)$
 is given in the form of Algorithm~\ref{aspalg}.
For simplicity of the presentation, we have only shown how to compute the costs stored in~$C_{G_{\sigma}}$, 
$\overline{\pmb{c}}_{G_{\sigma}}$ and~$\pmb{c}^*_{G_{\sigma}}$. The associated paths can be easily reconstructed by using pointers to the first and the last arcs of these paths, managed during the course of the algorithm.
\begin{algorithm}
\begin{small}
Perform parallel composition $G_{\sigma}\leftarrow p(G_{\mathrm{left}(\sigma)},G_{\mathrm{right}(\sigma)})$\;\label{parallel1}
$C_{G_{\sigma}}\leftarrow \min\{C_{G_{\mathrm{left}(\sigma)}},C_{G_{\mathrm{right}(\sigma)}}\}$\; \label{parallel6}
$\pmb{c}^*_{G_{\sigma}}[0]\leftarrow \min\{\pmb{c}^*_{G_{\mathrm{left}(\sigma)}}[0], \pmb{c}^*_{G_{\mathrm{right}(\sigma)}}[0]\}$\;
 \label{parallel2}
\For{$l=1$ \emph{\KwTo} $k$}{
$\pmb{c}^*_{G_{\sigma}}[l]\leftarrow \min\{\pmb{c}^*_{G_{\mathrm{left}(\sigma)}}[l], \pmb{c}^*_{G_{\mathrm{right}(\sigma)}}[l], 
C_{G_{\mathrm{left}(\sigma)}}+
\overline{\pmb{c}}_{G_{\mathrm{right}(\sigma)}}[l],  C_{G_{\mathrm{right}(\sigma)}}+\overline{\pmb{c}}_{G_{\mathrm{left}(\sigma)}}[l]\}$
\label{parallel4}\;
$\overline{\pmb{c}}_{G_{\sigma}}[l]\leftarrow \min\{\overline{\pmb{c}}_{G_{\mathrm{left}(\sigma)}}[l],\overline{\pmb{c}}_{G_{\mathrm{right}(\sigma)}}[l]\}$ \label{parallel5}
}
\Return{$(G_{\sigma}, \pmb{c}^*_{G_{\sigma}}, C_{G_{\sigma}}, \overline{\pmb{c}}_{G_{\sigma}})$}
  \caption{\small\texttt{P-incl}$(G_{\mathrm{left}(\sigma)},\pmb{c}^*_{G_{\mathrm{left}(\sigma)}},
  C_{G_{\mathrm{left}(\sigma)}},
  \overline{\pmb{c}}_{G_{\mathrm{left}(\sigma)}}, G_{\mathrm{right}(\sigma)}, \pmb{c}^*_{G_{\mathrm{right}(\sigma)}},C_{G_{\mathrm{right}(\sigma)}},\overline{\pmb{c}}_{G_{\mathrm{right}(\sigma)}})$}
 \label{parallel}
\end{small} 
\end{algorithm}

\begin{algorithm}
\begin{small}
Perform series composition $G_{\sigma}\leftarrow s(G_{\mathrm{left}(\sigma)},G_{\mathrm{right}(\sigma)})$\;
$C_{G_{\sigma}}\leftarrow C_{G_{\mathrm{left}(\sigma)}}+C_{G_{\mathrm{right}(\sigma)}}$\;\label{series4}
\For{$l=0$ \emph{\KwTo} $k$}{
$\pmb{c}^*_{G_{\sigma}}[l]\leftarrow \min_{0\leq j\leq l}\{\pmb{c}^*_{G_{\mathrm{left}(\sigma)}}[j]+\pmb{c}^*_{G_{\mathrm{right}(\sigma)}}[l-j]\}$\label{series3}
}
$\overline{\pmb{c}}_{G_{\sigma}}[1]\leftarrow +\infty$\;
\For{$l=2$ \emph{\KwTo} $k$}{
$\overline{\pmb{c}}_{G_{\sigma}}[l]\leftarrow \min_{1\leq j< l}\{\overline{\pmb{c}}_{G_{\mathrm{left}(\sigma)}}[j]+\overline{\pmb{c}}_{G_{\mathrm{right}(\sigma)}}[l-j]\}$\label{series7}
}   
\Return{$(G_{\sigma}, \pmb{c}^*_{G_{\sigma}}, C_{G_{\sigma}}, \overline{\pmb{c}}_{G_{\sigma}})$}
  \caption{\small \texttt{S-incl}$(G_{\mathrm{left}(\sigma)},\pmb{c}^*_{G_{\mathrm{left}(\sigma)}},
  C_{G_{\mathrm{left}(\sigma)}},
  \overline{\pmb{c}}_{G_{\mathrm{left}(\sigma)}}, G_{\mathrm{right}(\sigma)}, \pmb{c}^*_{G_{\mathrm{right}(\sigma)}},C_{G_{\mathrm{right}(\sigma)}},\overline{\pmb{c}}_{G_{\mathrm{right}(\sigma)}})$}
 \label{series}
\end{small} 
\end{algorithm}

\begin{algorithm}
\begin{small}
Find the binary decomposition tree~$\mathcal{T}$ of~$G$\;\label{aspalg0}
\ForEach(\tcp*[f]{Initially $\sigma$ corresponds to arc~$e\in A$ }){leaf~$\sigma$ of~$\mathcal{T}$ \label{aspalg01} }
{
$\pmb{c}^*_{G_{\sigma}}[0]\leftarrow C_e+\overline{c}_e$\;\label{aspalg1}
\lFor{$l=1$ \emph{\KwTo} $k$}{$\pmb{c}^*_{G_{\sigma}}[l]\leftarrow+\infty$}
$C_{G_{\sigma}}\leftarrow C_e$\;
$\overline{\pmb{c}}_{G_{\sigma}}[1]\leftarrow\overline{c}_e$\;
\lFor{$l=2$ \emph{\KwTo} $k$}{$\overline{\pmb{c}}_{G_{\sigma}}[l]\leftarrow+\infty$}\label{aspalg7}
}
\While{there is a node~$\sigma$ in $\mathcal{T}$ with two leaves:
 $\mathrm{left}(\sigma)$ and $\mathrm{right}(\sigma)$ \label{aspalg8}}{
 \eIf{node~$\sigma$ is  labeled~$\mathtt{P}$ in $\mathcal{T}$}{\label{aspalg9}
  $(G_{\sigma},\pmb{c}^*_{G_{\sigma}}, C_{G_{\sigma}}, \overline{\pmb{c}}_{G_{\sigma}})\leftarrow
   \texttt{P-incl}(G_{\mathrm{left}(\sigma)},\pmb{c}^*_{G_{\mathrm{left}(\sigma)}}C_{G_{\mathrm{left}(\sigma)}},
   \overline{\pmb{c}}_{G_{\mathrm{left}(\sigma)}}, G_{\mathrm{right}(\sigma)},
                                       \pmb{c}^*_{G_{\mathrm{right}(\sigma)}},C_{G_{\mathrm{right}(\sigma)}},
                                       \overline{\pmb{c}}_{G_{\mathrm{right}(\sigma)}} )$ \label{aspalgpar}
}(\tcp*[f]{node~$\sigma$ is labeled~$\mathtt{S}$ in $\mathcal{T}$})
{
    $(G_{\sigma},\pmb{c}^*_{G_{\sigma}}, C_{G_{\sigma}}, \overline{\pmb{c}}_{G_{\sigma}})\leftarrow
   \texttt{S-incl}(G_{\mathrm{left}(\sigma)},\pmb{c}^*_{G_{\mathrm{left}(\sigma)}}C_{G_{\mathrm{left}(\sigma)}},
   \overline{\pmb{c}}_{G_{\mathrm{left}(\sigma)}}, G_{\mathrm{right}(\sigma)},
                                       \pmb{c}^*_{G_{\mathrm{right}(\sigma)}},C_{G_{\mathrm{right}(\sigma)}},
                                       \overline{\pmb{c}}_{G_{\mathrm{right}(\sigma)}} )$     \label{aspalgser}
}
Delete   leaves $\mathrm{left}(\sigma)$ and $\mathrm{right}(\sigma)$  from~$\mathcal{T}$\;        \label{aspalg13}                          
 }
\Return{$\min_{0\leq l\leq k} \pmb{c}^*_{G_{\sigma}}[l]$}\tcp*[f]{node $\sigma$ is the root of~$\mathcal{T}$ corresponding to~$G$}
  \caption{Solve \textsc{Rec SP} in ASP$(G,\{C_e\}_{e\in A},\{\overline{c}_e\}_{e\in A},k,s,t)$ with~$\Phi^{\mathrm{incl}}(X,k)$}
   \label{aspalg}
\end{small} 
\end{algorithm}

An algorithm for \textsc{Rec SP} with~$\Phi^{\mathrm{excl}}(X,k)$ is  symmetric to the one with~$\Phi^{\mathrm{incl}}(X,k)$.
We also store three pieces of information for each ASP subgraph~$G_{\sigma}$. Namely,
the cost of a shortest path from  the source  to  the sink in~$G_{\sigma}$ under costs~$\overline{c}_e$, $e\in A$:
\begin{equation}
\overline{c}_{G_{\sigma}}=\min_{Y\in \Phi_{G_{\sigma}}} \sum_{e\in Y} \overline{c}_e.
\label{upcostexcl}
\end{equation}
The $k$-element array~$\pmb{C}_{G_{\sigma}}$, whose $l$-th element~$\pmb{C}_{G_{\sigma}}[l]$
is the cost of a shortest path from the source to the sink in~$G_{\sigma}$ under costs~$C_e$, $e\in A$,
that uses exactly~$l$ arcs if such a path exists:
\begin{equation}
\pmb{C}_{G_{\sigma}}[l]=
\begin{cases}
\displaystyle \min_{X\in \Phi^{l}_ {G_{\sigma}}} \sum_{e\in X}C_e &\text{if $ \Phi^{l}_ {G_{\sigma}}\not=\emptyset$,}\\
+\infty&\text{otherwise},
\end{cases}
\quad l=1,\ldots,k.
\label{locostexcl}
\end{equation}
The third element is the $(k+1)$-element array $\pmb{c}^*_{G_{\sigma}}$ whose definition is the same as~(\ref{defcs}) with the exception that $\Psi^{\mathrm{incl}[l]}_{G_{\sigma}}$ is replaced with $\Psi^{\mathrm{excl}[l]}_{G_{\sigma}}$.
The initial values  of $\overline{c}_{G_\sigma}$, $\pmb{C}_{G_{\sigma}}$ and $\pmb{c}^*_{G_{\sigma}}$
for each leaf node~$\sigma$ of the tree~$\mathcal{T}$
are as follows:
\begin{align}
&\overline{c}_{G_{\sigma}}= \overline{c}_e,&  \label{ini1excl}\\
&\pmb{C}_{G_{\sigma}}[1]=C_e,\; \pmb{C}_{G_{\sigma}}[l]=+\infty, &l=2,\ldots,k, \label{ini2excl}\\
&\pmb{c}^*_{G_{\sigma}}[0]=C_e+\overline{c}_e,\;\pmb{c}^*_{G_{\sigma}}[l]=+\infty, &l=1,\ldots,k. \label{ini3excl}
\end{align}
It is easily seen that adapting 
 Algorithms~\ref{parallel} and~\ref{series} and, in consequence, Algorithm~\ref{aspalg} for the 
 neighborhood~$\Phi^{\mathrm{excl}}(X,k)$
 requires small changes in them that take into account~$\overline{c}_{G_\sigma}$, $\pmb{C}_{G_{\sigma}}$, $\pmb{c}^*_{G_{\sigma}}$, and
(\ref{ini1excl})--(\ref{ini3excl}).

An algorithm for \textsc{Rec SP} with the symmetric difference neighborhood $\Phi^{\mathrm{sym}}(X,k)$ 
combines the ideas for the neighborhoods $\Phi^{\mathrm{incl}}(X,k)$ and $\Phi^{\mathrm{excl}}(X,k)$.
The three arrays are now associated with each ASP subgraph~$G_{\sigma}$. Namely, the $(k-1)$-element arrays~$\pmb{\overline{c}}_{G_{\sigma}}$, and $\pmb{C}_{G_{\sigma}}$  defined as~(\ref{upcost}) and~(\ref{locostexcl}), respectively, for $l\in [k-1]$.  The $(k+1)$-element array~$\pmb{c}^*_{G_{\sigma}}$, whose definition is the same as~(\ref{defcs}), with the exception that $\Psi^{\mathrm{incl}[l]}_{G_{\sigma}}$ is replaced with $\Psi^{\mathrm{sym}[l]}_{G_{\sigma}}$.
Of course, $\pmb{c}^*_{G_{\sigma}}[1]=+\infty$.
The initial values of $\pmb{\overline{c}}_{G_{\sigma}}$, $\pmb{C}_{G_{\sigma}}$, and $\pmb{c}^*_{G_{\sigma}}$,
for each leaf node~$\sigma$ of the tree~$\mathcal{T}$,
are as follows:
\begin{align}
&\overline{\pmb{c}}_{G_{\sigma}}[1]= \overline{c}_e,&\overline{\pmb{c}}_{G_{\sigma}}[l]=  +\infty, &&l=2,\ldots,k-1,
\label{ini1sym}\\
&\pmb{C}_{G_{\sigma}}[1]=C_e, &\pmb{C}_{G_{\sigma}}[l]=+\infty, &&l=2,\ldots,k-1, \label{ini2sym}\\
&\pmb{c}^*_{G_{\sigma}}[0]=C_e+\overline{c}_e,&\pmb{c}^*_{G_{\sigma}}[l]=+\infty, &&l=1,\ldots,k. \label{ini3sym}
\end{align}

\begin{algorithm}[h]
\begin{small}
Perform parallel composition $G_{\sigma}\leftarrow p(G_{\mathrm{left}(\sigma)},G_{\mathrm{right}(\sigma)})$\;
$\pmb{c}^*_{G_{\sigma}}[0]\leftarrow \min\{\pmb{c}^*_{G_{\mathrm{left}(\sigma)}}[0], \pmb{c}^*_{G_{\mathrm{right}(\sigma)}}[0]\}$\;
$\pmb{c}^*_{G_{\sigma}}[1]\leftarrow +\infty$\;
 \For{$l=2$ \emph{\KwTo} $k$}{
 $\pmb{c}^*_{G_{\sigma}}[l]\leftarrow \min_{1\leq j< l}
 \{\pmb{c}^*_{G_{\mathrm{left}(\sigma)}}[l], \pmb{c}^*_{G_{\mathrm{right}(\sigma)}}[l],   
 \pmb{C}_{G_{\mathrm{left}(\sigma)}}[j]+
\overline{\pmb{c}}_{G_{\mathrm{right}(\sigma)}}[l-j],  
\pmb{C}_{G_{\mathrm{right}(\sigma)}}[j]+\overline{\pmb{c}}_{G_{\mathrm{left}(\sigma)}}[l-j]\}$\;             
 }
\For{$l=1$ \emph{\KwTo} $k-1$}{
$\pmb{C}_{G_{\sigma}}[l]\leftarrow \min\{\pmb{C}_{G_{\mathrm{left}(\sigma)}}[l],\pmb{C}_{G_{\mathrm{right}(\sigma)}}[l]\}$\;
$\overline{\pmb{c}}_{G_{\sigma}}[l]\leftarrow \min\{\overline{\pmb{c}}_{G_{\mathrm{left}(\sigma)}}[l],\overline{\pmb{c}}_{G_{\mathrm{right}(\sigma)}}[l]\}$
}
\Return{$(G_{\sigma}, \pmb{c}^*_{G_{\sigma}}, \pmb{C}_{G_{\sigma}}, \overline{\pmb{c}}_{G_{\sigma}})$}
  \caption{\small \texttt{P-sym}$(G_{\mathrm{left}(\sigma)},\pmb{c}^*_{G_{\mathrm{left}(\sigma)}},
  \pmb{C}_{G_{\mathrm{left}(\sigma)}},
  \overline{\pmb{c}}_{G_{\mathrm{left}(\sigma)}}, G_{\mathrm{right}(\sigma)}, \pmb{c}^*_{G_{\mathrm{right}(\sigma)}},
  \pmb{C}_{G_{\mathrm{right}(\sigma)}},\overline{\pmb{c}}_{G_{\mathrm{right}(\sigma)}})$}
 \label{parallelsym}
\end{small} 
\end{algorithm}

\begin{algorithm}[h]
\begin{small}
Perform series composition $G_{\sigma}\leftarrow s(G_{\mathrm{left}(\sigma)},G_{\mathrm{right}(\sigma)})$\;
\For{$l=0$ \emph{\KwTo} $k$}{
$\pmb{c}^*_{G_{\sigma}}[l]\leftarrow \min_{0\leq j\leq l}\{\pmb{c}^*_{G_{\mathrm{left}(\sigma)}}[j]+\pmb{c}^*_{G_{\mathrm{right}(\sigma)}}[l-j]\}$
}
$\pmb{C}_{G_{\sigma}}[1]\leftarrow  +\infty$\;
$\overline{\pmb{c}}_{G_{\sigma}}[1]\leftarrow +\infty$\;
\For{$l=2$ \emph{\KwTo} $k-1$}{
$\pmb{C}_{G_{\sigma}}[l]\leftarrow \min_{1\leq j< l}\{\pmb{C}_{G_{\mathrm{left}(\sigma)}}[j]+\pmb{C}_{G_{\mathrm{right}(\sigma)}}[l-j]\}$\;
$\overline{\pmb{c}}_{G_{\sigma}}[l]\leftarrow \min_{1\leq j< l}\{\overline{\pmb{c}}_{G_{\mathrm{left}(\sigma)}}[j]+\overline{\pmb{c}}_{G_{\mathrm{right}(\sigma)}}[l-j]\}$
}   
\Return{$(G_{\sigma}, \pmb{c}^*_{G_{\sigma}}, \pmb{C}_{G_{\sigma}}, \overline{\pmb{c}}_{G_{\sigma}})$}
  \caption{\small \texttt{S-sym}$(G_{\mathrm{left}(\sigma)},\pmb{c}^*_{G_{\mathrm{left}(\sigma)}},
  \pmb{C}_{G_{\mathrm{left}(\sigma)}},
  \overline{\pmb{c}}_{G_{\mathrm{left}(\sigma)}}, G_{\mathrm{right}(\sigma)}, \pmb{c}^*_{G_{\mathrm{right}(\sigma)}},\pmb{C}_{G_{\mathrm{right}(\sigma)}},\overline{\pmb{c}}_{G_{\mathrm{right}(\sigma)}})$}
 \label{seriessym}
\end{small} 
\end{algorithm}
In order to give a version of  Algorithm~\ref{aspalg} for 
 the 
 neighborhood~$\Phi^{\mathrm{sym}}(X,k)$, we need to modify the initialization according to
 (\ref{ini1sym})--(\ref{ini3sym} (lines~\ref{aspalg01}--\ref{aspalg7}) and
 call Algorithms~\ref{parallelsym} and~\ref{seriessym} in the lines~\ref{aspalgpar} and~\ref{aspalgser}, respectively.
 
 We now prove the correctness of the algorithms.
\begin{thm}
\textsc{Rec SP} with~$\Phi^{\mathrm{incl}}(X,k)$, $\Phi^{\mathrm{excl}}(X,k)$ and $\Phi^{\mathrm{sym}}(X,k)$  in  
 arc series-parallel multidigraphs~$G=(V,A)$ can be solved in
 $O(|A|k^2)$ time.
\end{thm}
\begin{proof}
We will give a proof only for the arc inclusion neighborhood~$\Phi^{\mathrm{incl}}(X,k)$.
The proof for the neighborhood~$\Phi^{\mathrm{excl}}(X,k)$ is almost the same, while
a proof for~$\Phi^{\mathrm{sym}}(X,k)$ is similar in spirit to the proof given below. 
It is enough to make some technical changes to
take into account the differences among the neighborhoods described previously in this section.

We first prove that for each node~$\sigma$  of the decomposition tree~$\mathcal{T}$,
 graph~$G_{\sigma}$ is an ASP subgraph of~$G$ and the 
corresponding  costs~$C_{G_\sigma}$, $\overline{\pmb{c}}_{G_{\sigma}}$ and $\pmb{c}^*_{G_{\sigma}}$ 
are correctly computed by Algorithm~\ref{aspalg}. The first part of the claim is due to~\cite{VTL82}.
It is easily seen that after running  lines~\ref{aspalg1}--\ref{aspalg7}, for each leaf~$\sigma$ of~$\mathcal{T}$,
the costs $C_{G_\sigma}$, $\overline{\pmb{c}}_{G_{\sigma}}$ and $\pmb{c}^*_{G_{\sigma}}$, 
corresponding to the single-arc subgraphs, are correctly initialized. Hence, if $|A|=1$, then the claim trivially holds. Assume that $|A|\geq 2$. The
 proof is by induction on the number of iterations of  Algorithm~\ref{aspalg} (the lines~\ref{aspalg9}--\ref{aspalg13}).

The base case is when the algorithm runs for only one iteration ($|A|=2$).
The node~$\sigma$ is then the root of~$\mathcal{T}$ and has two  leaves
corresponding to 
two single-arc subgraphs~$G_{\mathrm{left}(\sigma)}$ and~$G_{\mathrm{right}(\sigma)}$.
If $\sigma$ is  labeled~$\mathtt{P}$, then the algorithm~ \texttt{P-incl} is called for 
$G_{\mathrm{left}(\sigma)}$ and~$G_{\mathrm{right}(\sigma)}$, otherwise the  \texttt{S-incl} is called.
Consider the  \texttt{P-incl} case. First the parallel composition of~$G_{\mathrm{left}(\sigma)}$ and~$G_{\mathrm{right}(\sigma)}$
   is performed. The resulting graph~$G_{\sigma}$
 has two parallel arcs, $e_1,e_2$, the source and the sink. We get
 $\pmb{c}^*_{G_{\sigma}}[0]=\min\{C_{e_1}+\overline{c}_{e_1},C_{e_2}+\overline{c}_{e_2}\}$,
 $\pmb{c}^*_{G_{\sigma}}[1]=\min\{+\infty,+\infty,C_{e_1}+\overline{c}_{e_2},C_{e_2}+\overline{c}_{e_1}\}$,
 $\overline{\pmb{c}}_{G_{\sigma}}[1]=\min\{\overline{c}_{e_1},\overline{c}_{e_2}\}$ and
 $C_{G_{\sigma}}=\min\{C_{e_1},C_{e_2}\}$.
 The rest of elements of $\pmb{c}^*_{G_{\sigma}}$ and $\overline{\pmb{c}}_{G_{\sigma}}$ are equal to $+\infty$.
 Let us turn to the case  \texttt{S-incl}. Now, after 
 the series composition of~$G_{\mathrm{left}(\sigma)}$ and~$G_{\mathrm{right}(\sigma)}$,
 the resulting  graph~$G_{\sigma}$ has three nodes $\{s,v_1,t\}$ and two series arcs $e_1=(s,v_1)$, $e_2=(v_1,t)$.
 Therefore, only one $s$-$t$ path exists in~$G_{\sigma}$. We get $\pmb{c}^*_{G_{\sigma}}[0]=C_{e_1}+\overline{c}_{e_1}+C_{e_2}+\overline{c}_{e_2}$,
 $C_{G_{\sigma}}=C_{e_1}+C_{e_2}$ and 
 $\overline{\pmb{c}}_{G_{\sigma}}[2]=\overline{c}_{e_1}+\overline{c}_{e_2}$.
 The rest of the elements of $\pmb{c}^*_{G_{\sigma}}$ and $\overline{\pmb{c}}_{G_{\sigma}}$ are equal to $+\infty$. Hence all the costs $C_{G_\sigma}$, $\overline{\pmb{c}}_{G_{\sigma}}$ and $\pmb{c}^*_{G_{\sigma}}$  are correctly computed for the base case. 
 
 For the induction step, let~$\sigma$ be an internal  node of~$\mathcal{T}$ with two leaves  $\mathrm{left}(\sigma)$ and $\mathrm{right}(\sigma)$.
 By the induction hypothesis, $G_{\mathrm{left}(\sigma)}$ and~$G_{\mathrm{right}(\sigma)}$ are ASP subgraphs and
the costs $C_{G_\sigma}$, $\overline{\pmb{c}}_{G_{\sigma}}$ and $\pmb{c}^*_{G_{\sigma}}$,  associated  with them,
 are correctly computed. We need to consider two cases that depend on the label of~$\sigma$.
 The first case (the label~$\mathtt{P}$) is when the algorithm  \texttt{P-incl} is called for 
$G_{\mathrm{left}(\sigma)}$ and~$G_{\mathrm{right}(\sigma)}$.
After the parallel composition of $G_{\mathrm{left}(\sigma)}$ and~$G_{\mathrm{right}(\sigma)}$,
the resulting subgraph~$G_{\sigma}$ is an ASP subgraph. 
Note that, $\Phi_{G_{\sigma}}= \Phi_{\mathrm{left}(\sigma)} \cup \Phi_{G_{\mathrm{left}(\sigma)}}$,
where $\Phi_{G_{\mathrm{left}(\sigma)}}\cap \Phi_{G_{\mathrm{left}(\sigma)}}=\emptyset$.
Hence, we immediately get that the costs computed in lines~\ref{parallel6},  \ref{parallel2}, and~\ref{parallel5} of Algorithm~\ref{parallel} are correctly computed.
The costs $\pmb{c}^*_{G_{\sigma}}[l]$, for $l=1,\ldots,k$, (see line~\ref{parallel5} of Algorithm~\ref{parallel}),
are correctly computed as well.
Indeed, the first two terms in the minimum are obvious. The last two terms take into account the cases when the
optimal paths $X^*\in \Phi_{G_{\sigma}}$ and $Y^*\in \Phi_{G_{\sigma}}$ in $\Psi^{\mathrm{incl}[l]}_{G_{\sigma}}$ are such that
$X^*\in \Phi_{G_{\mathrm{left}(\sigma)}}$ and $Y^*\in \Phi_{G_{\mathrm{right}(\sigma)}}$ or
$X^*\in \Phi_{G_{\mathrm{right}(\sigma)}}$ and $Y^*\in \Phi_{G_{\mathrm{left}(\sigma)}}$.
Consider the second case (the label~$\mathtt{S}$) when the algorithm~ \texttt{S-incl} is called.
The subgraph~$G_{\sigma}$, after  the series composition of $G_{\mathrm{left}(\sigma)}$ and~$G_{\mathrm{right}(\sigma)}$,
is an ASP subgraph and there is a node~$v$ in $G_{\sigma}$ being  the sink of~$G_{\mathrm{left}(\sigma)}$
and the source of~$G_{\mathrm{right}(\sigma)}$.
Thus, every path $X\in \Phi_{G_{\sigma}}$ must traverse the node~$v$ and $X$ is 
the concatenation of $X_1\in  \Phi_{G_{\mathrm{left}(\sigma)}}$ and 
$X_2\in  \Phi_{G_{\mathrm{right}(\sigma)}}$.
In what follows,
for each $l=0,\ldots,k$, the cost
$\pmb{c}^*_{G_{\sigma}}[l]$ in $G_{\sigma}$
is the sum of its optimal counterparts, respectively, in  $G_{\mathrm{left}(\sigma)}$ and~$G_{\mathrm{right}(\sigma)}$,
for some $0\leq j^*\leq l$.
It is enough to find such~$j^*$ for each~$l$.
A similar argument holds for the costs
$\overline{\pmb{c}}_{G_{\sigma}}[l]$ for $l=2,\ldots,k$.
Therefore, Algorithm~\ref{series} correctly computes the costs $C_{G_\sigma}$, $\overline{\pmb{c}}_{G_{\sigma}}$ and $\pmb{c}^*_{G_{\sigma}}$  in lines~\ref{series4}, \ref{series3} and~\ref{series7}. This proves the claim.

If $\sigma$  is  the root of~$\mathcal{T}$, then  $G_{\sigma}=G$. It immediately  follows  from the claim,
that  the array~$\pmb{c}^*_{G_{\sigma}}$ contains the optimal costs of the solutions in $\Psi^{\mathrm{incl}[l]}_{G_{\sigma}}$ for $l=0,\ldots,k$.
Thus, 
the cost 
of an optimal solution,~$X^*,Y^*\in \Phi$ to \textsc{Rec SP} in~$G$ is equal to $\min_{0\leq l\leq k} \pmb{c}^*_{G_{\sigma}}[l]$.

Let us analyze the running time of Algorithm~\ref{aspalg}. The binary
decomposition tree $\mathcal{T}$ of~$G$, line~\ref{aspalg0}, can be constructed in $O(|A|)$ time~\cite{VTL82}.
The initialization, lines~\ref{aspalg01}--\ref{aspalg7}, can be done in $O(|A|)$ time.
The root of $\mathcal{T}$ can be
reached in $O(|A|)$ time, lines~\ref{aspalg8}--\ref{aspalg13}. Algorithms~\ref{parallel} and~\ref{series}
require~$O(k^2)$ time.
Hence, the running time of Algorithm~\ref{aspalg} is $O(|A|k^2)$. Since 
the algorithm for $\Phi^{\mathrm{excl}}(X,k)$ is symmetric, its
running time is the same.
Observe that  Algorithms~\ref{parallelsym} and~\ref{seriessym}  require~$O(k^2)$ time.
Thus, the running time of algorithm for $\Phi^{\mathrm{sym}}(X,k)$  is $O(|A|k^2)$.
\end{proof}

\section{The \textsc{Rec Rob SP}  problem under the budgeted uncertainty}
\label{secrobust}

In this section, we consider the \textsc{Rec Rob SP}  problem under the budgeted uncertainty. Adding budgets to the interval uncertainty representation makes the problem more challenging. Of course, all the hardness results for 
\textsc{Rec SP} (see Section~\ref{srsppm}) remain valid for \textsc{Rec Rob SP}  with the budgeted uncertainty. In particular, the problem is strongly NP-hard and not approximable in general digraphs for all the neighborhoods under consideration.
Table~\ref{tab2} summarizes the main results obtained later in this section.
\begin{table}[h]
 \begin{small}
 \caption{Summary of the main results for  \textsc{Rec Rob SP}
 under the budgeted uncertainty
 for various neighborhoods.}
 \label{tab2}
 \begin{center}
  \begin{tabular}{l|l|l|l}
                      \hline
   Budgeted            &\multicolumn{3}{c}{Neighborhoods}\\
                     \cline{2-4}
uncertainty & $\Phi^{\mathrm{incl}}(X,k)$ & $\Phi^{\mathrm{excl}}(X,k)$ & $\Phi^{\mathrm{sym}}(X,k)$\\ 
\hline \hline
    $\mathcal{U}(\Gamma^c)$    & MIP&  MIP &  MIP \\
                                                   \cline{2-4}
                                                  & para-NP-complete&&\\ 
    \hline
     $\mathcal{U}(\Gamma^d)$    & MIP(?) - open problem& $\Sigma_3^p$-hard~\cite{JKZ24} & $\Sigma_3^p$-hard~\cite{JKZ24} \\
                                      \cline{2-4}
                                     & strongly NP-hard       &    $\Sigma^p_3$-hard& $\Sigma^p_3$-hard\\
                                      &not approximable&  to approximate& to approximate~\cite{JKZ24}\\
                                      & for $k=\Gamma^d=1$~\cite{NO13}&for $k=2$~\cite{JKZ24}&\\
  \hline
 \end{tabular}
 \end{center}
  \end{small}
  \end{table}

\subsection{The \textsc{Rec Rob SP}  problem under~$\mathcal{U}(\Gamma^c)$}

We now show a compact MIP
formulation for \textsc{Rec Rob SP}  under the scenario set $\mathcal{U}(\Gamma^c)$. We will use the fact that $\mathcal{U}(\Gamma^c)$ is a special case of polyhedral uncertainty.

\begin{thm}
\label{tmcfuc}
The \textsc{Rec Rob SP}  under~$\mathcal{U}(\Gamma^c)$ in  general digraph $G=(V,A)$, $|A|=m$,
 with neighborhoods $\Phi^{\mathrm{incl}}(X,k)$, $\Phi^{\mathrm{excl}}(X,k)$ and $\Phi^{\mathrm{sym}}(X,k)$,
admits a compact MIP formulation.
\end{thm}
\begin{proof}
The idea of the construction is similar to that in~\cite[Theorem~1]{BG22a} (see also~\cite{GLW24}),
where a compact MIP formulation for a recoverable version of a single-machine scheduling problem under uncertainty was proposed. It uses  Carath{\'e}odory's theorem,  which states that any point in a convex hull ${\rm Conv}(P)$ of a set $P\subset \mathbb{R}^n$ can be expressed as a convex combination of at most $n+1$ points in $P$. Then, using the minimax theorem (see, e.g.,~\cite{BG22a} for details) allows us to express \textsc{Rec Rob SP}  as follows:
\begin{equation}
\label{robrecequiv}
\min_{\underset{Y^{(1)},\ldots, Y^{(m+1)} \in \Phi(X,k)}{X\in \Phi}} \left(\sum_{e\in X} C_e + \max_{S\in \mathcal{U}(\Gamma^c)}
\min_{i\in [m+1]}
\sum_{e\in Y^{(i)}} c_e^S\right).
\end{equation}
Fix $\pmb{x}\in \chi(\Phi)$ and  $\pmb{y}^{(i)}\in \chi(\Phi(\pmb{x},k))$, $i\in [m+1]$ (see Section~\ref{secmiprec} for the description of the sets $\chi(\Phi)$ and $\chi(\Phi(\pmb{x},k))$. The inner problem
 \begin{equation}
 \label{advincl}
 \max_{S\in \mathcal{U}(\Gamma^c)}
\min_{i\in [m+1]}
\sum_{e\in Y^{(i)}} c_e^S
\end{equation}
 can be modeled as follows:
 \begin{align}
           \max \;& t& \label{advincl1}\\
                  & t \leq  \sum_{e\in A}  \hat{c}_e y^{(i)}_e+  \sum_{e\in A}  y^{(i)}_e u_e & \forall i\in [m+1], \label{advincl2}\\
     &0\leq u_e\leq \Delta_e & \forall e\in A, \label{advincl3}\\
     & \sum_{e\in A} u_e\leq  \Gamma^c.&\label{advincl4}
\end{align}
Dualizing the inner problem, we get
 \begin{align*}
     \min \;&\sum_{ i\in [m+1]} \left(\sum_{e\in A}  \hat{c}_e y^{(i)}_e\right)\lambda_i + 
      \sum_{e\in A}   \Delta_e \gamma_e + \Gamma^c\Theta&\\
      &\sum_{ i\in [m+1]} \lambda_i=1,&\\
     &\gamma_e+\Theta\geq \sum_{ i\in [m+1]}  y^{(i)}_e \lambda_i  & \forall e\in A,\\
     & \lambda_i \geq 0 &\forall i\in [m+1],\\
     &\gamma_e \geq 0 &\forall e\in A,\\
     &\Theta\geq 0.&
\end{align*}
By the strong duality, \textsc{Rec Rob SP}  under~$\mathcal{U}(\Gamma^c)$ can be rewritten as
\begin{align}
     \min \;\;&\sum_{e\in A} C_e x_e+ \sum_{e\in A}  \hat{c}_e  \sum_{ i\in [m+1]}y^{(i)}_e\lambda_i + 
      \sum_{e\in A}   \Delta_e \gamma_e + \Gamma^c\Theta& \label{ncel1}\\
      &\sum_{ i\in [m+1]} \lambda_i=1,&\label{ncel2}\\
     &\gamma_e+\Theta\geq \sum_{ i\in [m+1]}  y^{(i)}_e \lambda_i  & \forall e\in A, \label{ncel3}\\
     & \lambda_i \geq 0 &\forall i\in [m+1],\\
     &\gamma_e \geq 0 &\forall e\in A,\\
     &\Theta\geq 0,&\\
      &\pmb{y}^{(i)}\in \chi(\Phi(\pmb{x},k)) & \forall i\in [m+1], \label{ncel7}\\
           &\pmb{x}\in \chi(\Phi).& \label{ncel8}
\end{align}
The constraint~(\ref{ncel8}) can be replaced with~(\ref{chx0})--(\ref{chx4}).
The constraints~(\ref{ncel7}) can be replaced with one of (\ref{cneighincl})--(\ref{cneighsym}), depending on the neighborhood used (see Section~\ref{secmiprec}). Finally, it suffices to linearize the terms $y_e^{(i)}\lambda_i$ that appear in~(\ref{ncel3}). This can be done by substituting $v_i=y_e^{(i)}\lambda_i$, $i\in [m+1]$, and adding the constraints $v_i\leq M y_e^{(i)}$, $v_i\geq \lambda_i-M(1-y_e^{(i)})$, $v_i\leq \lambda_i$, $v_i\geq 0$.

\end{proof}

The following theorem establishes the parameterized complexity hardness of~\textsc{Rec Rob SP}, not only when parameterized by~$k$ (see Theorem~\ref{thminclp}), but also
by~$\Gamma^c$.
\begin{thm} 
The decision version of~\textsc{Rec Rob SP}  under~$\mathcal{U}(\Gamma^c)$ 
  with~$\Phi^{\mathrm{incl}}(X,k)$ parameterized by $k$ and~$\Gamma^c$
is para-NP-complete in general digraphs.
\label{tincugc}
\end{thm}
\begin{proof}
The problem~\textsc{Rec Rob SP}  is in  NP, since (\ref{robrecequiv}) shows that
$X\in \Phi$ and 
$Y^{(1)},\ldots, Y^{(m+1)} \in \Phi^{\mathrm{incl}}(X,k)$ is
a certificate for proving if an instance of~\textsc{Rec Rob SP}  is a Yes-instance and 
computing the value of $\sum_{e\in X} C_e + \max_{S\in \mathcal{U}(\Gamma^c)}
\min_{i\in [m+1]}
\sum_{e\in Y^{(i)}} c_e^S$ can be done in polynomial time by solving the
linear programming model (\ref{advincl1})--(\ref{advincl4}).
Thus, \textsc{Rec Rob  SP} is in para-NP.
Since \textsc{Rec Rob  SP} is NP-hard for constant~$k$ and~$\Gamma^c$ ($k=\Gamma^c=1$)~\cite{NO13},
 it is para-NP-hard.
\end{proof}
Theorem~\ref{tincugc} shows that \textsc{Rec Rob SP}  under~$\mathcal{U}(\Gamma^c)$ 
  with~$\Phi^{\mathrm{incl}}(X,k)$ parameterized by $k$ and~$\Gamma^c$
is para-NP-hard in general digraphs.

\subsection{The \textsc{Rec Rob SP}  problem under~$\mathcal{U}(\Gamma^d)$} 
It has been recently  proven in~\cite{JKZ24} that \textsc{Rec Rob SP} with scenario set $\mathcal{U}(\Gamma^d)$ and neighborhoods $\Phi^{\mathrm{excl}}(X,k)$, $\Phi^{\mathrm{sym}}(X,k)$ is $\Sigma^p_{3}$-hard
and  $\Sigma^p_{3}$-hard to approximate
 in general digraphs. This excludes a compact MIP formulation for the problem unless the polynomial-time 
 hierarchy collapses (see, e.g.,~\cite{WO21}).
 Construction of a compact MIP formulation for the neighborhood $\Phi^{\mathrm{incl}}(X,k)$, or for acyclic digraphs, is an interesting open problem.

It was shown in~\cite{NO13} that \textsc{Rec Rob SP}   with~$\Phi^{\mathrm{incl}}(X,k)$
under~$\mathcal{U}(\Gamma^d)$
 is strongly 
NP-hard and not approximable in general digraphs, even
if the recovery parameter~$k=1$ and the budget~$\Gamma^d=1$,  by
a reduction  from $K$-\textsc{V-DP},  where the number of terminal pairs~$K=2$, for a general digraph~\cite{FHW80}.
Since $K$-\textsc{V-DP} remains  strongly
NP-complete in an acyclic planar digraph if $K$ is unbounded~\cite{NS09},
we can easily obtain a reduction from $K$-\textsc{V-DP} in an acyclic planar digraph, 
which is a minor generalization of the one given in \cite{NO13}.
This allows us 
to show that \textsc{Rec Rob SP}  under~$\mathcal{U}(\Gamma^d)$ with~$\Phi^{\mathrm{incl}}(X,k)$
is strongly NP-hard and not approximable
if $k=\Gamma^d=1$ in digraphs with a simpler structure than general
digraphs, namely, nearly acyclic planar ones.
From the above, we immediately get that 
\textsc{Rec Rob SP}  under~$\mathcal{U}(\Gamma^d)$ with~$\Phi^{\mathrm{incl}}(X,k)$
is para-NP-hard in these simpler digraphs.

\subsection{Approximation algorithms}

In this section, we again study the \textsc{Rec Rob SP}  problem in acyclic multidigraphs. We have shown in Section~\ref{secinterv}  that this problem can be solved in polynomial time for the interval uncertainty set $\mathcal{U}$. We now use this fact to provide some approximation algorithms for the problem under scenario sets $\mathcal{U}(\Gamma^c)$ and $\mathcal{U}(\Gamma^d)$, whose computational complexity status is unknown in acyclic multidigraphs.  Let 
$$
F(X)= \left(\sum_{e\in X} C_e + \max_{S\in \Uset}\min_{Y\in \Phi(X,k)} \sum_{e\in Y} c_e^S\right),
$$
where $\mathbb{U}\in \{\mathcal{U}(\Gamma^c),\mathcal{U}(\Gamma^d)\}$ and $\Phi(X,k)$ is one of the neighborhoods~(\ref{incl})--(\ref{sym}).
Fix scenario $S\in \mathbb{U}$ and consider the problem
$$
\textsc{Rec SP $(S)$}:\; \min_{X\in \Phi, Y\in \Phi(X,k)} \left(\sum_{e\in X} C_e +  \sum_{e\in Y} c_e^S\right).
$$

The idea behind the construction of the approximation algorithms is to take into account the cost structure in scenario sets $\mathcal{U}(\Gamma^c)$ and $\mathcal{U}(\Gamma^d)$. Let $\alpha\geq (0,1]$ be a constant such that $\hat{c}_e\geq \alpha(\hat{c}_e+\Delta_e)$ for each $e\in A$. This condition means that the maximum second-stage cost of any arc $e\in A$ can be at most $\frac{1}{\alpha}$ greater than its nominal cost.   In practical applications, it can be unlikely that the value of $\alpha$ is very small as the maximum increase in the costs is typically limited. Notice that such a constant exists if $\hat{c}_e>0$ for each $e\in A$. The following result has been established in~\cite{HKZ16a} for the recoverable robust spanning tree problem (its proof for \textsc{Rec Rob SP} is almost verbatim):
\begin{lem}{\cite[Lemma 6]{HKZ16a}}
\label{lemappr1}
Suppose that $\hat{c}_e\geq \alpha (\hat{c}_e+\Delta_e)$ for each $e\in A$, where $\alpha \in (0,1]$, and let $(\hat{X},\hat{Y})$ be an optimal solution to \textsc{Rec SP $(S)$} for $S=(\hat{c}_e)_{e\in A}$. Then $F(\hat{X})\leq \frac{1}{\alpha}F(X)$ for each $X\in \Phi$.
\end{lem}

Lemma~\ref{lemappr1} is true for both uncertainty sets $\mathcal{U}(\Gamma^d)$ and $\mathcal{U}(\Gamma^d)$. Let us now focus only on $\mathcal{U}(\Gamma^c)$.
Let $D=\sum_{e\in A} \Delta_e>0$ be the total deviation of the second-stage arc costs from their nominal values. The value of $D$ can be seen as a total uncertainty in the second-stage arc costs. Define scenario $S'$ in which $c_e^{S'}=\min\{\hat{c}_e+\Delta_e, \hat{c}_e+\Gamma^c\frac{\Delta_e}{D}\}$, $e\in A$. It is easy to see that $S'\in \mathcal{U}(\Gamma^c)$, since $\sum_{e\in A} (c_e^{S'}-\hat{c}_e)\leq \sum_{e\in A} \Gamma^c \frac{\Delta_e}{D}=\Gamma^c$. Let $(\hat{X},\hat{Y})$ be an optimal solution to \textsc{Rec SP $(S')$}. Suppose that we  can find two constants $\beta\in(0,1]$ and $\gamma\in [0,1)$ such that $\beta D\leq \Gamma^c\leq \gamma F(\hat{X})$. The constants $\beta$ and $\gamma$ relate the budget $\Gamma^c$ to the total deviation $D$ and the total cost of the heuristic solution $\hat{X}$.  The following lemma has been proven in~\cite{HKZ16a} for the recoverable robust spanning tree problem (again, its proof for \textsc{Rec Rob SP} is almost verbatim):
\begin{lem}{\cite[Lemma 7]{HKZ16a}}
\label{lemappr2}
  Assume that $\mathbb{U}=\mathcal{U}(\Gamma^c)$ and let $(\hat{X},\hat{Y})$ be an optimal solution to \textsc{Rec~SP~($S'$)}. Then, the following implications are true:
  \begin{enumerate}
    \item[(a)] if $\Gamma^c\geq \beta D$ for $\beta\in (0,1]$, then $F(\hat{X})\leq \frac{1}{\beta}F(X)$ for each $X\in \Phi$,
    \item[(b)] if $\Gamma^c \leq \gamma F(\hat{X})$ for $\gamma \in [0,1)$, then  $F(\hat{X})\leq \frac{1}{1-\gamma} F(X)$ for each $X\in \Phi$.
  \end{enumerate}
\end{lem}
Observe that \textsc{Rec SP (S)} can be solved in polynomial time in acyclic multidigraphs, according to the results shown in Section~\ref{secinterv} (it is equivalent to \textsc{Rec SP} with $\overline{c}_e=c^S_e$, $e\in A$).
The value of $F(X)$ for a given $X\in \Phi$ 
can be computed in polynomial time for scenario set $\mathcal{U}(\Gamma^c)$ and 
the neighborhood~$\Phi^{\mathrm{incl}}(X,k)$, see~\cite{NO13}, where the time-expanded network approach~\cite{AMO93}
was applied. The same technique can be applied to computing 
the value of $F(X)$  for~$\mathcal{U}(\Gamma^c)$ and 
the neighborhoods~$\Phi^{\mathrm{excl}}(X,k)$ and~$\Phi^{\mathrm{sym}}(X,k)$ in polynomial time
in acyclic multidigraphs.
From this, it follows that one can
 verify the condition~$(b)$ from Lemma~\ref{lemappr2} in polynomial time. 
Therefore, Lemma~\ref{lemappr1} and Lemma~\ref{lemappr2} imply the following approximation results, where $\alpha$, $\beta$ and $\gamma$ are the constants from these lemmas:
\begin{cor}
  The \textsc{Rec Rob SP}  in acyclic multidigraphs  is approximable within $\frac{1}{\alpha}$ for scenario set $\mathcal{U}(\Gamma^d)$ and within $\min \{\frac{1}{\alpha}, \frac{1}{\beta},\frac{1}{1-\gamma}\}$ for scenario set $\mathcal{U}(\Gamma^c)$.
\end{cor}

\section{Conclusion}

In this paper, we have discussed the recoverable robust shortest path problem under various interval scenario sets and neighborhoods. The main results are polynomial time algorithms for acyclic multidigraphs and the traditional interval uncertainty. In addition, we have strengthened the known hardness results for general multidigraphs and provided some new results for the budgeted interval uncertainty. In particular, we have proposed a compact MIP formulation for the continuous budgeted uncertainty and some approximation algorithms for acyclic multidigraphs.
There are still several open questions regarding the problem being considered. Perhaps the most interesting one is the characterization of the complexity of the recoverable robust problem in acyclic multidigraphs under budgeted uncertainty. The problem is known to be hard only for general multidigraphs. A polynomial time algorithm may exist in general acyclic multidigraphs or for some of their special cases (layered, arc series-parallel, etc.). 
It is also interesting to provide a more detailed characterization of the problem from the parameterized complexity point of view.

\subsubsection*{Acknowledgements}
The authors were supported by
 the National Science Centre, Poland, grant 2022/45/B/HS4/00355.


\appendix

\section{Appendix}
\label{dod}

We show two examples that demonstrate that adding the anti-cycling constraints (\ref{chx1})--(\ref{chx4}) to the MIP formulation is necessary, even if all first and second-stage arc costs are positive.

\begin{figure}[ht]
\includegraphics[width=15.5cm]{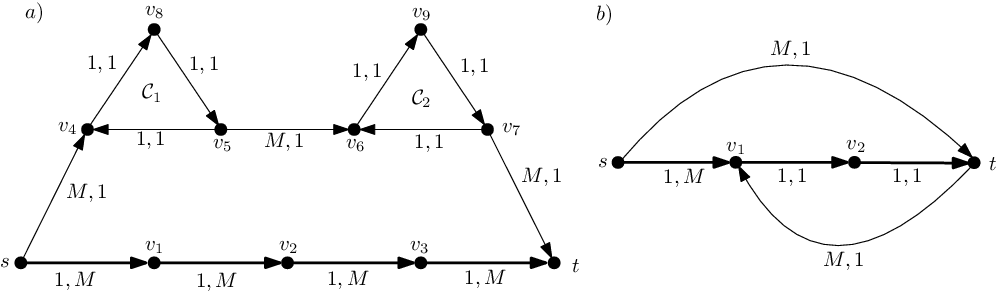}
\caption{Two sample graphs~$G=(V,A)$ with the first and second-stage arc costs
$C_e,\overline{c}_e$, $e\in A$, respectively, where
 $M$ is a big constant. The paths in bold are the optimal first-stage paths.} \label{figex1}
\end{figure}

Consider the sample digraph shown in Figure~\ref{figex1}a. Choose the arc inclusion neighborhood and the recovery parameter $k=3$. The total cost of any solution to this problem is at least $M$. Indeed, to achieve the cost less than $M$ we have to choose the path $X=\{(s,v_1),(v_1,v_2),(v_2,v_3),(v_3,t)\}$ as the first-stage path. However, achieving the cost of the second-stage path less than $M$ requires including more than~3 new arcs. Suppose that we remove the anti-cycling constraints (\ref{chx1})--(\ref{chx4}) from the description of $\chi(\Phi)$. Then $X'=X\cup \mathcal{C}_1\cup \mathcal{C}_2$, where 
$\mathcal{C}_1$ and $\mathcal{C}_2$ are directed cycles depicted in Figure~\ref{figex1}a, is a feasible first-stage solution to the MIP formulation. Its cost equals 10, which is less than $M$. We can now obtain a feasible second-stage solution $Y=\{(s,v_4),(v_4,v_8),(v_8,v_5), (v_5,v_6), (v_6,v_9), (v_9,v_7), (v_7,t)\}$ with the cost equal to~7 (observe that only 3 new arcs are added to $X'$). The total cost of the solution obtained is 17, which is less than $M$. Thus,
this example demonstrates that we have to add the anti-cycling constraints to the description of $\chi(\Phi)$.

Consider the sample digraph shown in Figure~\ref{figex1}b. Choose the arc-exclusion neighborhood and the recovery parameter $k=1$. The total cost of any solution to this problem is at least $M$. Indeed, to achieve a less cost we have to chose $X=\{(s,v_1), (v_1,v_2), (v_2,t)\}$ as the first-stage path. The cost of this path equals~3. However, the cost of the second-stage path $Y$ is then at least $M$ because each simple $s$-$t$ path in the neighborhood of $X$ must contain the arc $(s,v_1)$. If we remove the anti-cycling constraints from the description of $\chi(\Phi^{\rm excl}(\pmb{x},k))$, then we can choose $Y=\{(s,t), (t,v_1), (v_1,v_2), (v_2,v_1), (v_1,t)\}$, which is now feasible and has the cost equal to~4.  Observe that the path $Y$ is not simple. Similar reasoning applies to the arc symmetric difference neighborhood. It is enough to set the recovery parameter $k=3$. This example shows that we have to add the anti-cycling constraints to the description of $\chi(\Phi^{\rm excl}(\pmb{x},k))$ and $\chi(\Phi^{\rm sym}(\pmb{x},k))$.

\end{document}